\documentclass[letterpaper, 10 pt, conference]{ieeeconf}

\IEEEoverridecommandlockouts                              % This command is only needed if 
\usepackage[utf8]{inputenc}

\let\labelindent\relax
\usepackage{amssymb,amsmath,amsthm,mathrsfs,mathtools,upgreek}
\usepackage[dvipsnames]{xcolor}
\usepackage{soul}
\usepackage{times}
\usepackage{graphicx}
\usepackage{dblfloatfix}
\usepackage{caption}
\usepackage{subcaption}
\usepackage[shortlabels]{enumitem}
\setlist{nolistsep}
\usepackage{makecell}

\usepackage{algpseudocode}
\usepackage{algorithm}
\usepackage{cite}
\newtheorem*{proposition*}{Proposition}
\newtheorem*{theorem*}{Theorem}
\newtheorem{assumption}{Assumption}
\newtheorem{lemma}{Lemma}
\newtheorem{theorem}{Theorem}
\newtheorem{definition}{Definition}
\newtheorem{problem}{Problem}
\newtheorem{remark}{Remark}

\newtheorem{proposition}{Proposition}

\usepackage{tikz}
\usetikzlibrary{arrows}
\usetikzlibrary{automata,arrows,positioning,calc}
\usepackage{circuitikz}
\usetikzlibrary{shapes}
\definecolor{green1}{rgb}{0.2,0.7,0.2}

\newcommand{\PP}{\mathbb{P}}

\newcommand{\Os}{\mathcal{O}}

\title{\LARGE \bf Large Population Games on Constrained Unreliable Networks}

\author{Shubham~Aggarwal, Muhammad~Aneeq~uz~Zaman, Melih Bastopcu, \textit{Member, IEEE}, and Tamer~Ba{\c s}ar, \\ \textit{Life Fellow, IEEE}
\thanks{Research of the authors was supported in part by the ARO MURI Grant AG285 and in part by the AFOSR Grant FA9550-19-1-0353.}
\thanks{The authors are affiliated with the Coordinated Science Lab, University of Illinois at Urbana-Champaign, Urbana, IL, USA 61801. Emails:
        {\tt\small \{sa57,mazaman2,bastopcu,basar1\}@illinois.edu}.}%
}
\tikzset{ remember picture,
   switch/.style = {rectangle,
                    draw,align=center,
                    label={below:#1},
   },
}

\newsavebox\mybox
\savebox\mybox{%
\tikz\draw[line width=0.7pt] (-0.4,0)--(0,0)
								(0,0)--(0.4,0.4);%
}
\begin{document}
\tikzstyle{rect} = [draw,rectangle,fill = white!20,minimum width = 3pt, inner sep  = 5pt]
\tikzstyle{line} = [draw, -latex]
\tikzstyle{dline} = [draw, dash dot, -latex]

\maketitle
\thispagestyle{empty}

\begin{abstract}
This paper studies an $N$--agent cost-coupled game where the agents are connected via an \emph{unreliable capacity constrained} network. %addresses the effect of scheduling agents over unreliable networks on the Nash equilibrium in multi-agent dynamic games. 
Each agent receives state information over that network which loses packets with probability $p$. A Base station (BS) actively schedules agent communications over the network by minimizing a weighted Age of Information (WAoI) based cost function under a capacity limit $\mathcal{C} < N$ on the number of transmission attempts at each instant. Under a standard information structure, we show that the problem can be decoupled into a \emph{scheduling problem} for the BS and a \emph{game problem} for the $N$ agents. Since the scheduling problem is an NP hard combinatorics problem, we propose an approximately optimal solution which approaches the optimal solution as $N \rightarrow \infty$. In the process, we also provide some insights on the case without channel erasure. Next, to solve the large population game problem, we use the mean-field game framework to compute an approximate decentralized Nash equilibrium. %by employing the obtained scheduling policy. 
Finally, we validate the theoretical results using a numerical example.
\end{abstract}

\section{Introduction}
With the phenomenal expansion in data-traffic galvanized by the growing number of connected devices, Internet-of-Things (IoT) finds applications in diverse areas such as smart grids, autonomous vehicles, and monitoring systems \cite{nanda2019internet,rana2017distributed,osseiran20165g}, to name a few. A commonality among all of the above is the presence of distributed sensing and actuating devices communicating via a wireless network. While distributed systems can efficiently handle the growing network size compared to their centralized counterparts, they come with added challenges, such as limited channel capacities, network unreliability, and scalability concerns. These constraints might cause end-to-end latency or in a worse case, missing information at the end-user, which can lead to compromised reliability in safety-critical applications. Thus, there is an urgent need for the development of dependable and timeliness-aware communication technologies with the potential to mitigate the above posed concerns. In this work, we aim to propose strategies to mitigate the deleterious effects of unreliable capacity-constrained communication in networks involving a large number of decision-making agents.

Specifically, we consider a large population setting where $N$ rational agents aim to form consensus while communicating intermittently over a network. This intermittency is caused by i) a capacity-constrained downlink connecting the BS to the decoders, and ii) the possibility of \emph{erasure} amidst transmission, after information is relayed by the BS. This results in an unreliable capacity-constrained network. As a result, the agents must maintain an estimate of their state to consequently compute control actions that can achieve consensus. Meanwhile, the BS, which  is tasked with the scheduling of information, must carefully design policies to account for the heterogeneity in agent dynamics whilst also dealing with the possibility of erasure of the scheduled information. We formulate the BS's problem by proposing a Weighted Age of Information (WAoI) based cost function which is monotonically increasing in the average estimation error of the agents, thereby extending the setting of our earlier work \cite{aggarwal2022weighted} to erasure channels. Further, we improve upon the convergence guarantees in \cite{aggarwal2022weighted} for the case where the network is erasure free by proposing a novel scheduling policy. Finally, we employ this policy to construct an approximate Nash solution for the finite-agent consensus problem. 

In literature, the early works \cite{imer2006optimal, schenato2007foundations} have dealt with an optimal control problem with unreliable communication, albeit, for a single agent system and an unconstrained network under the TCP and the UDP communication protocols. The work \cite{moon2014discrete} extends the setting to multi-agent games; however, the considered network is unconstrained. In order to measure timeliness in communication networks, age of information (AoI) has been introduced as a potential metric. In the context of networked feedback systems, the AoI-based policies have been proposed for solving resource allocation and end-user uncertainty reduction problems as in \cite{Ayan2020}. %Additionally, it has been extended to the setting age of incorrect information (AoII) for solving multi-agent remote state estimation problems \cite{maatouk2020age1}. 
Recently, age of incorrect information (AoII) is proposed for solving multi-agent remote state estimation problems \cite{maatouk2020age1}. Age-optimal scheduling policies have been considered with Markovian error-prone channel state in \cite{chen2021optimizing, sombabu2022whittle}, with unknown erasure probabilities in \cite{Wu2022}, and over erroneous broadcast channels in \cite{Kadota18a}. %Age-optimal scheduling policies have been studied in \cite{chen2021optimizing} by utilizing Markov decision processes (MDP) for a Markovian error-prone channel state, in \cite{sombabu2022whittle}  by using Whittle's Index based scheduling for the Gilbert-Elliot channel, in \cite{Wu2022} by using online Whittle's index based policy over channels with unknown erasure probabilities, and finally in \cite{Kadota18a} by developing low-complexity scheduling policies for a broadcast channel over erasure channels.
A more detailed literature review on age-optimal scheduling policies can be found in \cite{Yates20a}.   

To appropriately handle the concerns of increasing network interactions, one of the most relevant framework is that of mean-field games (MFGs) \cite{huang2007large,lasry2007mean}. It leads one to circumvent the issues posed by scalability, by allowing for a representative agent to play against the population, although, at the cost of entailing an approximate equilibrium solution to the finite-agent consensus problem. It has been well-studied in the regime of linear-quadratic systems \cite{zaman2022reinforcement,aggarwal2022linear,bensoussan2016linear} and holds great potential to solve problems involving ultradense networks or massive machine-type communication \cite{zhang2021age,bennis2018ultrareliable}. 
For additional literature on large multi-agent systems with networked communication, we refer the reader to \cite{aggarwal2022weighted}.

We list below the main contributions of this paper. %First, we show that under standard information structure, the $(N+1)$--player problem ($N$ agents and a BS) can be separated into (a) a scheduling problem for the BS, and (b) a consensus game problem for the $N$--agents. 
We extend the setting of our previous works \cite{aggarwal2023ACC,aggarwal2022weighted} to the case of unreliable downlink communication. Since the scheduling problem belongs to the class of restless multi-armed bandits, for which an optimal policy is hard to compute, we propose a novel suboptimal maximum age-based tie-breaking protocol (MATB-P) to solve the capacity-constrained scheduling problem of the BS (which is also different from the uniform sampling-based policy considered in \cite{aggarwal2022weighted}). We prove that this policy approaches optimality (exponentially fast) as $N$ grows large, in contrast to the $\mathcal{O}(N^{-0.5})$ rate proposed in \cite{aggarwal2022weighted}. Further, we also provide high-probability guarantees on the tail of the AoI, which, in turn, provides guarantees on the freshness of information under high traffic. Additionally, in the special case with no channel erasure, we relax the assumption on the $A$ matrix in the work \cite{aggarwal2022weighted} by proving a uniform upper bound on the AoI of all agents under MATB-P. Finally, using the policy constructed above, we solve the $N$--agent consensus problem by leveraging the MFG paradigm and getting $\epsilon$--Nash policies for the agents where $\epsilon\xrightarrow{N \rightarrow\infty}0$.

The rest of the paper is organized as follows. We formulate the $(N+1)$--player game problem in Sec. \ref{sec:formulation}. In Sec.\ref{Sol:BSlevel}, we solve the BS-level scheduling problem, and provide its analysis in Sec. \ref{sec:Large_dev}. Then, we solve the agent-level game problem in Sec. \ref{sec:MFG_summ}, and provide a numerical example in Sec. \ref{sec:example}. The paper is concluded in Sec. \ref{sec:conclusion} with some major highlights, followed by four appendices, providing detailed derivations and proofs.

\textbf{Notations:} We let $[N]:= \{1,2, \cdots, N\}$ and $tr(\cdot)$ denote the trace of its argument matrix. The Euclidean 2-norm and the Frobenius norm are denoted by $\|\cdot\|$ and $\|\cdot\|_F$, respecively. All the empty summations are set to 0. For a vector $x$ and a positive semi-definite matrix $Q$, $\|x\|^2_Q:=x^\top Qx$. We define the limit superior of a real sequence as $\overline{\lim}:=\limsup$. %For real functions, $\operatorname{F}(\alpha)=\mathcal{O}(\operatorname{G}(\alpha))$ is equivalent to the statement that there exists $\operatorname{M} >0$ such that $\lim_{\alpha \rightarrow \infty}\frac{|\operatorname{F}(\alpha)|}{|\operatorname{G}(\alpha)|} = \operatorname{M}$. 
Finally, $\mathbf{1}_{A}$ denotes the indicator function of the argument.

\begin{figure}[t]
	\centerline{\includegraphics[width=1\columnwidth]{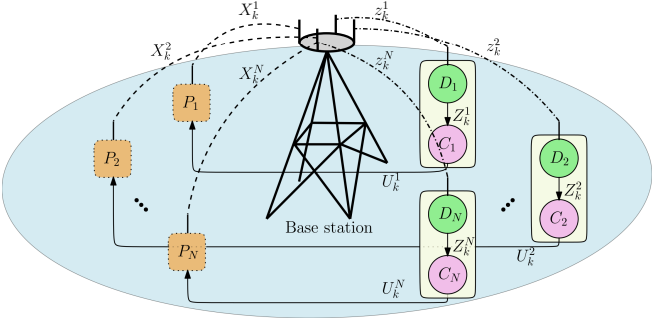}}
	\caption{\small{A prototypical networked control system constituting a BS and $N$ game playing agents. The BS, decoders, and controllers are active decision makers. Dashed lines denote an erasure-free wireless transfer, dotted-dashed lines denote erasure-prone one, and bold lines denote wired information transfer.}}
	\vspace{-0.2cm} 
	\label{Fig:Inf_flow}
    \vspace{-0.5cm} 
\end{figure}

\section{Problem Formulation}\label{sec:formulation}
In this section, we set up the two sub-problems in the $(N+1)$--player game, namely, a) the agent-level game problem, and b) the BS-level scheduling problem.

Consider a multi-agent system consisting of $N$ cost-coupled agents receiving information over an unreliable network. Each agent $i$ constitutes a plant, a decoder and a controller, labeled as a tuple $(P_i,D_i,C_i)$ as shown in Fig. \ref{Fig:Inf_flow}. Dynamics of $P_i$ evolve in discrete-time as
\begin{align}\label{system}
    X^i_{k+1} = A(\phi_i)X^i_k + B(\phi_i)U^i_k + W^i_k, ~k \geq 0,
\end{align}
where $X^i_k \in \mathbb{R}^n$ is the state and $U^i_k \in \mathbb{R}^m$ is the control input, both for agent $i$. The exogeneous noise $W^i_k \in \mathbb{R}^n$ is zero mean with covariance $C_W(\theta_i) > 0$. The initial state  $X^i_0$ of agent $i$ is assumed to have symmetric density with mean $x_{\phi_i,0}$ and covariance $\Sigma_x > 0$. Further, it is assumed to be independent of the noise process for all timesteps $k$. The system matrices $A(\phi_i),$ $B(\phi_i)$ are time-invariant with suitable dimensions. Further, they are chosen according to an empirical function $\mathbb{P}^N(\phi = \phi_i),$ where $\phi \in \Phi$ denotes the type of an agent chosen from a finite set $\Phi:= \{\phi_1, \cdots, \phi_p\}$. We assume that $|\mathbb{P}^N(\phi) - \mathbb{P}(\phi)| = \mathcal{O}(1/N),$ $\forall \phi$, where $\mathbb{P}(\phi)$ denotes the limiting distribution.

The state of plant $P_i$ is relayed to the decoder $D_i$ via an ideal uplink to the BS, which then regulates agent communications over the downlink. The downlink is constrained by a capacity limit of $\mathcal{C} < N$ units on the number of transmissions and serves as a bottleneck from the plant to the decoder. Further, it is unreliable in the sense that a packet communicated over it may be lost according to a Bernoulli distributed signal $\beta_k \sim \text{Ber}(p)$, with $p$ being the erasure probability. The decoder receives the information signal: 
\begin{align}
    z^i_k:= X^i_k \mathbf{1}_{E_k^i} + \emptyset \mathbf{1}_{(E_k^i)^c},
\end{align}
where the event $E^i_k$ denotes that state information is \textit{successfully} transmitted. Further, the event $\left(E^i_k\right)^c$ denotes no transmission (or $\emptyset$), which can be either due to no transmission by the BS or a packet drop over the channel. Let us denote the instants of information reception by the decoder as $\wp^i_k:=\zeta^i_k \beta_k$. Then, the information history of the decoder is defined as $I^{D_i}_k:= \{z^i_{0:k},\wp^{i}_{0:k},U^i_{0:k-1}\}$, based on which it computes the minimum mean-squared (MMS) estimate ($\mathbb{E}[X^i_k \mid I^{D_i}_k]$) of the state $X^i_k$. We adopt the convention that $z^i_{-1}=Z^i_{-1}=U^i_{-1}=0$, and $W^i_{-1}  = X^i_0- Z^i_0$, for all $i$.

Next, each controller $C_i$ receives the estimate from $D_i$ and aims to minimize the average cost function
\begin{align}\label{Finite_cost}
J_i(\pi_c)\!:= \!\overline{\lim\limits_{T\rightarrow \infty}} \frac{1}{T}\mathbb{E}\left[ \sum_{k=0}^{T-1} \!\!\|X^i_k \!-\! \mu^N_k\|_{Q(\phi_i)} \!+\! \|U^i_k\|_{R(\phi_i)} \right]\!\!,\!
\end{align}
where $Q(\phi_i) \geq 0$, $R(\phi_i)>0$, and $\mu^N_k := \frac{1}{N}\sum_{j=1}^NX^j_k$ represents the coupling between agents. %As a result of this term, the cost (and hence the minimizing policy) of an agent $i$ depends on not only its own but also on the policies of the other agents. 
Due to this coupling, the cost $J_i$ depends on the strategy $\pi_c:= \{\pi^1_c, \cdots, \pi^N_c\}$ of the entire population. In the sequel, we will denote the policy of the population excluding that of agent $i$ as $\pi^{-i}_c$. Further, $\pi^i_c \in \Pi_i:= \{\pi^i_c \mid \pi^i_c \text{ is adapted to } \sigma(I^{C_i}_s, s= 0, \cdots, k)\}$, $\forall i$, where $I^{C_i}_k:= \{U^j_{0:k-1},Z^j_{0:k}\}_{j \in [N]}$ denotes the information history of $C_i$, and $\sigma(\cdot)$ denotes the sigma-algebra generated by its argument. We assume that the pair $(A(\phi_i),B(\phi_i))$ is controllable and the pair $(A(\phi_i), \sqrt{Q(\phi_i)})$ is observable \cite{lewis2012optimal}. %We note that information history of the controller is a centralized one since it depends on the signals of the other agents. However, as is typically the case with networked games, which involve a large number of agents, having access to and keeping track of huge amount of information is difficult. Hence, 
Due to the difficulty in computing Nash equilibrium for the game \eqref{system}-\eqref{Finite_cost}, we will resort to the MFG framework (later in Section \ref{sec:MFG_summ}) to compute decentralized $\epsilon$-Nash policies where only local information will be required for decision-making and $\epsilon \rightarrow 0$ as $N \rightarrow \infty$. Next, %we proceed towards describing 
we describe the BS-level problem, where the objective is to compute an optimal scheduling policy of the BS.

% \subsection{BS-level Problem}
The aim of the BS is to \textit{efficiently} transmit information over the downlink. To this end, consider the most recent timestep when information was received by the $i^{th}$ controller, which is defined as $\ell^i_k:= \sup_{\ell \leq k}\{\ell \geq 0 \mid z^i_{\ell} \neq \emptyset\}$. Then, the AoI at the controller, which is the time elapsed since the generation of the most recent packet at the plant, is defined as $\tau^i_k:= k - \ell^i_k$. Further, its evolution is given as $\tau^i_{k+1} = (\tau^i_k+1)\mathbf{1}_{\{\wp^i_k=0\}}$,
% \begin{align}
% \tau^i_{k+1} = \begin{cases}
% 0, & \text{if} ~~\wp^i_k = 1, \\
% \tau^i_k + 1, & \text{otherwise},
% \end{cases}
% \end{align}
i.e., the AoI drops to zero only when a transmission is attempted by the BS and the packet is not dropped by the network. With the above AoI evolution, we formally define the capacity-constrained scheduling problem at the BS as follows:
\begin{problem}\label{Prob:Cap_Cons}
 \vspace{-0.5em}
 % \begin{small}
\begin{align*}
    \inf_{\gamma \in \Gamma}& J^{BS}(\gamma):= \overline{\lim}_{T \rightarrow \infty} \frac{1}{T} \mathbb{E} \left[ \frac{1}{N} \sum_{k=0}^{T-1}\sum_{i=1}^N w^i_k \tau^i_k \right] \\
     \mbox{s.t.}&~ \sum_{i=1}^N{\zeta^i_k} \leq \mathcal{C},~\forall k,
\end{align*}
% \end{small}
where $\gamma:= \{\gamma^1, \cdots, \gamma^N\}$ and $\Gamma:= \{\gamma \mid \gamma \text{ is adapted to}$  $\sigma(I^{BS}_s), s = 0, \cdots, k\}$ is the space of admissible scheduling policies with $I^{BS}_k:= \{\tau^i_{0:k},\zeta^i_{0:k-1},\wp^i_{0:k-1}\}_{i \in [N]}$ being the information history of the BS. Moreover, $w^i_k:= \mathbb{E}[\|e^i_k\|^2]$ denote the importance weights associated to each agent and are functions of the estimation error $e^i_k:= X^i_k - Z^i_k$ at the controller. Finally, the expectation is taken over the probabilistic scheduling due to the erasure-prone downlink and (possible) randomization in the scheduling policy.
\end{problem}

We note here that the information history of the BS includes the information reception instants of the agent decoders. This can be easily facilitated by a TCP-like protocol \cite{imer2006optimal}, where the decoder sends a one-bit ACK/NACK information to acknowledge whether or not the transmitted information was received by it. Further, Problem~\ref{Prob:Cap_Cons} involves a hard-limit on the number of transmissions, which makes it a combinatorics problem. It belongs to the class of restless multi-armed bandit problems, computing an optimal policy for which is quite difficult. Thus, in the sequel, we first reformulate the problem using the AoIs of each agent and then solve a relaxed problem involving a time-averaged constraint. The solution to the latter problem will then lead to a sub-optimal policy for Problem \ref{Prob:Cap_Cons}, which we will finally show to approach the optimal policy as $N$ increases.

To this end, we start by defining the shorthands $A_i:=A(\phi_i)$, $B_i:=B(\phi_i)$ and $C_{W^i}:=C_W(\phi_i)$. Then, we construct the decoder's MMS estimate as
\begin{align}
    Z^i_k = X^i_k \mathbf{1}_{[\wp^i_k=1]} + \mathbb{E}[X^i_k \mid I^{D_i}_k] \mathbf{1}_{[\wp^i_k=0]},
\end{align}
which upon using \eqref{system}, yields
\begin{align*}
    Z^i_k = X^i_k \mathbf{1}_{[\wp^i_k=1]} \!+ \!(A_iZ^i_{k-1} \!+\! B_iU^i_{k-1} \!+\! \mathbb{E}_c[W^i_{k-1}]) \mathbf{1}_{[\wp^i_k=0]},
\end{align*}
where $\mathbb{E}_c[\cdot] := \mathbb{E}[\cdot \mid \zeta^i_k = 0]$. Then, using similar arguments as in \cite{aggarwal2022weighted}, we can show that the term $\mathbb{E}_c[W^i_{k-1}]=0$ under the assumption of symmetric densities of $X^i_0$ and $W^i_0$. Hence, the estimate at the decoder can be easily computed as:
\begin{align}\label{Decoder_state}
     Z^i_k = X^i_k \mathbf{1}_{[\wp^i_k=1]} + (A_iZ^i_{k-1} + B_iU^i_{k-1}) \mathbf{1}_{[\wp^i_k=0]}.
\end{align}
With the above estimate, we can re-express the term $w^i_k$ in Problem \ref{Prob:Cap_Cons} using Lemma 1 from \cite{aggarwal2022weighted} as: %the following lemma.
% \begin{lemma}\cite{aggarwal2022weighted}
%     The estimation error of an agent $i$ is independent of all its past control inputs, and hence can be recast in terms of the AoI as $e^i_k = \sum_{\ell=1}^{\tau^i_k}A_i^\ell W^i_{k-\ell} \mathbf{1}_{[\wp^i_k = 0]}$. Moreover, we have that
    \begin{align}
        w^i_k := w^i_k(\tau^i_k,A_i,C_{W^i})=\sum_{\ell=1}^{\tau^i_k}tr\left({A^{\ell-1}_i}^\top A^{\ell-1}_i C_{W^i}\right).
    \end{align}
     % \end{lemma}
Now, let us define the running cost $c(\tau^i_k,A_i,C_{W^i}):= w^i_k\tau^i_k$. Then, since the capacity constraint in Problem \ref{Prob:Cap_Cons} makes the optimal policy difficult to compute, we relax the problem to one with an average constraint:

\begin{problem}\label{Prob:Relaxed_cons}
% \begin{small}
\begin{align*}
    \inf_{\gamma \in \Gamma}& J^{BS}(\gamma):= \overline{\lim}_{T \rightarrow \infty} \frac{1}{T} \mathbb{E} \left[ \frac{1}{N} \sum_{k=0}^{T-1}\sum_{i=1}^N c(\tau^i_k,A_i,C_{W^i}) \right] \\
    \mbox{s.t.} &~ \overline{\lim}_{T \rightarrow \infty} \frac{1}{T}\mathbb{E}\left[\sum_{k=0}^{T-1}\sum_{i=1}^N{\zeta^i_k}\right] \leq \mathcal{C}.
\end{align*}
% \end{small}
\end{problem}
We note that the constraint in the the above problem entails that more than $\mathcal{C}$ agents can be connected over the downlink at any given timestep as long as the capacity constraint is satisfied in the long run. This is clearly a weaker constraint than the one in Problem \ref{Prob:Cap_Cons} since the latter requires the capacity constraint to be satisfied at all timesteps $k$. Hence, it is indeed a relaxation of Problem \ref{Prob:Cap_Cons}. The objective now is to compute an optimal solution to Problem \ref{Prob:Relaxed_cons} and then utilize the solution to come up with an asymptotically optimal solution to Problem \ref{Prob:Cap_Cons}. To this end, we start by constructing the Lagrangian of Problem \ref{Prob:Relaxed_cons} (with $\lambda \geq 0$ the Lagrange multiplier):
\begin{align*}
    \mathscr{L}(\gamma,\lambda) \!=\! \overline{\lim\limits_{T \rightarrow \infty}}\frac{1}{T} \mathbb{E} \!\!\left[ \frac{1}{N}\!\! \sum_{k=0}^{T-1}\sum_{i=1}^N c(\tau^i_k,A_i,C_{W^i})\!\! +\! \lambda \!\!\left(\zeta^i_k \!\!- \!\!\frac{\mathcal{C}}{N} \right)\!\! \right]\!\!,
\end{align*}
where $ \lambda$ can be thought of as a price on the downlink utilization. Thus, given a fixed $\lambda$, we decouple the $N$--agent scheduling problem into $N$ decoupled single-agent problems:
\begin{problem}\label{decoupled_problem} For all $i \in [N]$,
\begin{align*} 
     \inf_{\gamma^i \in \Gamma^i} V^i(\gamma):=  \overline{\lim}_{T \rightarrow \infty} \frac{1}{T} \mathbb{E} \left[ \sum_{k=0}^{T-1} c(\tau^i_k,A_i,C_{W^i})  +\lambda \zeta^i_k \right].
\end{align*}
\end{problem}
In the next section, we will first solve Problem \ref{decoupled_problem}, for which we will cast the evolution of the AoI in an MDP framework and then construct a suboptimal policy for Problem \ref{Prob:Cap_Cons}. Additionally, we will also suppress the superscript $i$.
\section{Solution to the BS-level Problem}\label{Sol:BSlevel}
We compute an optimal policy for Problem \ref{decoupled_problem} by first defining it as a discrete-time MDP $\operatorname{M}:=(\operatorname{S},\operatorname{A},\operatorname{P},\operatorname{C})$. The state space $\operatorname{S}$ is the space of non-negative integers. The action set $\operatorname{A}= \{0,1\}$. An action $a=0$ denotes that a transmission is not attempted while $a=1$ denotes that it is. The probability transition function $\operatorname{P}$ describes the evolution of the AoI, i.e., $\operatorname{P}(\tau_{k+1} = 0 \mid \tau_k) = a_k(1-p)$ and $\operatorname{P}(\tau_{k+1} = \tau_k+1 \mid \tau_k) = 1-a_k +a_kp$. Finally, with the per stage cost defined to be $C(\tau,a):= c(\cdot) + \lambda a$, the MDP objective is to infimize the function $V(\gamma):= \overline{\lim}_{T \rightarrow \infty}\frac{1}{T}\mathbb{E} \big[\sum_{k=0}^{T-1} \operatorname{C}(\tau_k,a_k) \big]$ 
% \azedit{Having defined Problem \ref{decoupled_problem} as an MDP we solve it below.} {\color{gray}Then, we can formally define the MDP objective as
% \begin{align}\label{Prob:MDP}
%     \inf_{\gamma \in \Gamma} V(\tau,\gamma):= \overline{\lim\limits_{T \rightarrow \infty}}\frac{1}{T}\mathbb{E} \left[\sum_{k=0}^{T-1} \operatorname{C}(\tau_k,a_k) \right],
% \end{align}
for which we compute an optimal policy next. %%in the next subsection.
% \subsection{Single-Agent Scheduling Policy}
\subsection{Solution to Problem \ref{Prob:Relaxed_cons}}
\noindent We start by stating the following theorem, which characterizes an optimal policy solving Problem~\ref{decoupled_problem}.
\begin{theorem}\label{Th:determ_scheduling}
    Given $\lambda \geq 0$, there exists a stationary policy $\gamma_s$ solving the above MDP with an optimal cost of $\sigma^*$, which is independent of $\tau$. Moreover, the optimal policy is given as $a:= \mathbf{1}_{[\tau \geq \kappa]}$, for integer $\kappa:= \kappa(A,C_W,\lambda)$.
\end{theorem}
% \begin{proof}[Proof Sketch]
%    The proof of the theorem relies of first constructing a discounted-cost MDP on an infinite horizon and arguing that the optimal policy for the same has a threshold structure. Then, using the main results of \cite{sennott1993constrained}, we can show threshold structure of the policy as in the statement of the theorem. The details follow in a similar manner as in \cite{aggarwal2022weighted}.
% \end{proof}
The proof follows in a similar manner as the proof of \cite[Theorem~1]{aggarwal2022weighted}. %Now that we have established 
The theorem says that the optimal policy for Problem \ref{decoupled_problem} %\eqref{Prob:MDP} 
is a threshold policy. Next, we
%the remaining piece is to 
compute the threshold parameter $\kappa$ by invoking the following condition, which links the erasure probability $p$ and the instability in the agent's dynamics.
\begin{assumption}\label{As:sys_param}
We have that $\|A\|_F^2p < 1$.
\end{assumption}
Notice that the above assumption is standard in the literature on unreliable communication \cite{imer2006optimal,schenato2007foundations,moon2014discrete} and formalizes the fact that a higher erasure probability restricts our ability to stabilize highly unstable agents. In the extreme case when no communication is possible (i.e., $p = 1$) it requires %instates 
that all agents must be stable. The detailed derivation for computing $\kappa$ is provided in Appendix~\ref{ap:Compute_kappa}.
%Finally, we note that the structural Assumption \ref{As:sys_param} on the system parameters is very similar in spirit to that proposed in literature \cite{imer2006optimal,schenato2007foundations,moon2014discrete} and serves to couple the degree of instability of an agent that can be stabilized with a given probability of transmission failure. 
Now, with the deterministic single-agent policy as provided above, we proceed toward constructing an optimal policy for Problem \ref{Prob:Relaxed_cons}
% \subsection{Multi-agent Scheduling Policy}
% Now that we have explicitly characterized the optimal policy structure for Problem \ref{decoupled_problem} given a fixed $\lambda$, the aim in this subsection is to construct an optimal policy for Problem \ref{Prob:Relaxed_cons}. 
which, as we will see, will be a randomized policy since the optimal policy for such a constrained optimization problem may not, in general, lie in the class of stationary deterministic policies \cite{altman1999constrained}. Henceforth, we resume the use of superscript $i$  to denote the $i^{th}$ agent.

 We start by computing an optimal value of $\lambda$. To this end, consider the threshold parameter $\kappa^i(\lambda):= \kappa^i(A_i,C_{W^i},\lambda)$ as in Theorem 
 \ref{Th:determ_scheduling}. Then, the expected return time of agent $i$ starting from $\tau^i_k = 0$ can be found by $\operatorname{R}^i_{0} = [\sum_{r=0}^\infty(\kappa^i+r+1)(1-p)^{r+1}p^r]^{-1}$ and is equal to:
%
% \begin{align}
%     \operatorname{R}^i_{0} &= \frac{1}{\sum_{r=0}^\infty(\kappa^i+r+1)(1-p)^{r+1}p^r} \nonumber \\
      % &= \frac{\kappa(1-p)}{p} + \frac{1-p}{p^2}
%      & = \frac{((1-p)p-1)^2}{(1-p)((p-1)p(\kappa^i+1)+(1-p)p + \kappa^i + 1)}.
 %\end{align}
  \begin{align}
     \operatorname{R}^i_{0} \!= \!\frac{((1-p)p-1)^2}{(1-p)((p-1)p(\kappa^i+1)\!+\!(1-p)p \!+\! \kappa^i \!+\! 1)}.\!
 \end{align}
Then, under the average constraint in Problem \ref{Prob:Relaxed_cons}, we have $R(\lambda):= \sum_{i=1}^N\operatorname{R}^i_{0} \leq \mathcal{C}$. Consequently, we can use the iterative Bisection search algorithm, as given in \cite{aggarwal2022weighted,maatouk2020age1}, starting with the initial parameters $\underline{\lambda}^{(0)} = 0$, and $\overline{\lambda}^{(0)} = 1$. The algorithm terminates when $|\overline{\lambda}^{(m)} -\underline{\lambda}^{(m)}|\leq \epsilon$, for an iterating index $m$ and a suitably chosen $\epsilon > 0.$ 
% {\color{gray}which we summarize as follows. We initialize two parameters $\underline{\lambda}^{(0)} = 0$ and $\overline{\lambda}^{(0)} = 1$. We then calculate the threshold parameters $\kappa^i(\lambda_u^{(0)})$ for all $i$, by using \eqref{kappa_computation} and \eqref{kapp_computation1}. Consequently, we iterate by setting $\underline{\lambda}^{(j+1)} = \overline{\lambda}^{(j)}$ and $\overline{\lambda}^{(j+1)} = 2\overline{\lambda}^{(j)}$ until the constraint \eqref{update_rate_constraint} is satisfied for $\overline{\lambda}^{(r)}$, for some integer $r$. Then, we define the interval $[\underline{\lambda}^{(r)},\overline{\lambda}^{(r)}]$ which contains the optimal value of the multiplier $\lambda^*$, that can be calculated using the \emph{Bisection method}. The iteration stops when $|\overline{\lambda}^{(m)} -\underline{\lambda}^{(m)}|\leq \epsilon$, for the iterating index $m$ and for a suitably chosen $\epsilon > 0.$ }
Next, let us define $\underline{\lambda}^* = \underline{\lambda}^{(m)}$ and $\overline{\lambda}^*= \overline{\lambda}^{(m)}$ as obtained above, and the corresponding deterministic policies as $\gamma^{i}_{s_1}$ and $\gamma^{i}_{s_2}$, which are obtained from Theorem \ref{Th:determ_scheduling}. More precisely, we have that $\underline{\lambda}^*\!\! \mapsto \!\underline{\kappa}(\underline{\lambda}^*)\!\!:= \!\!\{\underline{\kappa}^1(\underline{\lambda}^*), \cdots, \underline{\kappa}^N(\underline{\lambda}^*)\}^\top\!\!$ and $\overline{\lambda}^*\!\! \mapsto\!$ $ \overline{\kappa}(\overline{\lambda}^*):=\!\{\overline{\kappa}^1(\overline{\lambda}^*),\cdots, \overline{\kappa}^N(\overline{\lambda}^*)\}^\top$. Also, let $\overline{\mathcal{C}}$ and $\underline{\mathcal{C}}$ be the total capacities used corresponding to the multipliers $\overline{\lambda}^*$ and $\underline{\lambda}^*$, respectively. Then, we define %the probability $q$ and 
the deterministic policies: %, for all $i$ as:
% \begin{subequations}\label{Random_policy}
% \begin{align}
% % \label{Prob_of_randomization} q & := \frac{\mathcal{C} - \overline{\mathcal{C}}}{\underline{\mathcal{C}} - \overline{\mathcal{C}}}, \\
% \label{deterministic_policy1}\gamma^{i}_{s_1}(\tau^i) & := \mathbf{1}_{[\tau^i \ge \underline{\kappa}^i(A_i,C_{W^i},\underline{\lambda}^*)]}, \\
% \label{deterministic_policy2}\gamma^{i}_{s_2}(\tau^i) & := \mathbf{1}_{[\tau^i \ge \overline{\kappa}^i(A_i,C_{W^i},\overline{\lambda}^*)]},
% \end{align}
% \end{subequations}
\begin{align}\label{Random_policy}
\gamma^{i}_{s_1}(\tau^i) := \mathbf{1}_{[\tau^i \ge \underline{\kappa}^i(\cdot,\cdot,\underline{\lambda}^*)]}, \hspace{0.15cm}
\gamma^{i}_{s_2}(\tau^i) := \mathbf{1}_{[\tau^i \ge \overline{\kappa}^i(\cdot,\cdot,\overline{\lambda}^*)]},
\end{align}
for all $i$ using which we can construct a randomized policy $\gamma_R:= [\gamma^{1}_{R}, \cdots, \gamma^{N}_{R}]^\top$ for the relaxed Problem \ref{Prob:Relaxed_cons} as:
 \begin{small}
\begin{align}\label{Randomized_policy}
\hspace{-0.2cm} \gamma^{i}_R = q \gamma^{i}_{s_1} + (1-q) \gamma^{i}_{s_2}, ~\forall i, %\frac{\mathcal{C} - \overline{\mathcal{C}}}{\underline{\mathcal{C}} - \overline{\mathcal{C}}}
\end{align}
\end{small}
where $q := (\mathcal{C} - \overline{\mathcal{C}})/(\underline{\mathcal{C}} - \overline{\mathcal{C}})$ is the probability of randomization. Next, in the following proposition, we state that the randomized policy obtained is indeed optimal for Problem~\ref{Prob:Relaxed_cons}.

\begin{proposition}\cite{aggarwal2022weighted}\label{Optimality_of_Randomized_policy}
Under Assumption \ref{As:sys_param}, the policy \eqref{Random_policy}-\eqref{Randomized_policy} is optimal for the relaxed minimization Problem \ref{Prob:Relaxed_cons}. 
\end{proposition}

With the solution to Problem~\ref{Prob:Relaxed_cons}, in the next subsection, we propose a novel asymptotically optimal policy for Problem~\ref{Prob:Cap_Cons}.

% \subsection{Capacity-constrained Scheduling Policy}\label{subsec:hard_cons_pol}
\vspace{-0.1em}
\subsection{Solution to Problem \ref{Prob:Cap_Cons}}\label{subsec:hard_cons_pol}
In this subsection, we provide a sub-optimal solution to Problem \ref{Prob:Cap_Cons}, using the solution to Problem \ref{Prob:Relaxed_cons}, which is shown to be asymptotically optimal as $N \rightarrow \infty$. We refer to this policy as the maximum-age-first tie-breaking protocol (or MATB-P for short). %To construct the same, c
Consider the solution $\gamma^i_R$ to Problem \ref{Prob:Relaxed_cons} as computed in the previous subsection and let $a^i_k = \gamma^i_R(I^{BS}_k)$ be the scheduling action at timestep $k$. Define $\Lambda_k:= \{j \in [N] \mid a^j_k = 1\}$ as the set of agents scheduled to be transmitted at instant $k$ and its cardinality to be $n^\lambda_{k}$. Then, the scheduling decision $\zeta^i_k$ under MATB-P ($\gamma^i$) is given as:
\begin{itemize}
    \item If $n^\lambda_k \leq \mathcal{C}$, then $\zeta^i_k = a^i_k$
    \item If $n^\lambda_k > \mathcal{C}$, then ${\zeta}^{i}_k = 1$ for a subset $\Lambda^{max}_k \subset \Lambda_k$ of the agents, where the cardinality of $\Lambda^{max}_k$ is $\mathcal{C}$ for all $k$, and it constitutes the agents with the maximum values of $\tau_k$. The agents in the set $\Lambda_k\setminus \Lambda^{max}_k$ remain unselected.
\end{itemize}
In the next section, we provide a tail-bound analysis of the constructed MATB policy, first, for the special case with no channel erasure, and then, for the general case.

\section{Tail-bound Analysis \& $\varepsilon$-Optimality}\label{sec:Large_dev}
In this section, we show that the costs under $\gamma_R$ and ${\gamma}$ approach each other as $N \rightarrow \infty$. To this end, we first prove Proposition \ref{Prop:Delta_boundedness} for the case of an ideal downlink with $p=0$, where we show that the maximum AoI is uniformly bounded independent of $N$, and then Theorem \ref{Prop:Delta_High_deviation_bounded} for the general non-ideal downlink case, where we provide a high confidence bound on the maximum AoI, again, independent of $N$. Then, we finally show (using Theorems \ref{Th:Asymptotic_optimality_deterministic} and \ref{Th:erasure_system}) that $\gamma$ approaches the optimal policy as $N \rightarrow \infty$ in both cases.

To this end, consider the Markov chain induced by the relaxed policy $\gamma^{i}_R$ for the $i^{th}$ agent as
\begin{align*}
\tau^i_{k+1} \!=\! \left\{ \begin{array}{ll}
\tau^i_k + 1, \qquad \text{w.p.}~ 1, &  \tau^i_k < \underline{\kappa}^i(\underline{\lambda}^*), \\
\!\!\left\{ \begin{array}{ll}
\tau^i_k + 1, & \text{w.p.}~ (1-q)p, \\
{0}, &\text{w.p.}~ 1-(1-q)p,
\end{array}
\right.\!\! & \!\!\tau^i_k \!=\!\underline{\kappa}^i(\underline{\lambda}^*), \\
\!\!\left\{ \begin{array}{ll}
\tau^i_k + 1, & \text{w.p.}~ p, \\
{0}, &\text{w.p.}~ 1-p,
\end{array}
\right. & \tau^i_k  \geq \overline{\kappa}^i(\overline{\lambda}^*).
\end{array}
\right.
\end{align*}
Then, since each state in the set $\operatorname{S}$ is reachable from every other state, the above Markov chain is irreducible, and hence admits a unique stationary distribution $\pi^i$. Now, we provide the following proposition which shows that the AoI under MATB-P for a deterministic channel (with $p=0$) is uniformly bounded, independent of $N$. 

\begin{proposition}\label{Prop:Delta_boundedness}
Under a fixed $\alpha \!= \!\mathcal{C}/N$ and $p = 0$, the AoI $\tau^i_k$ of any agent $i \in [N]$  under MATB-P is bounded by $\Os(\alpha^{-1})$.
\end{proposition}
The proof can be found in Appendix \ref{Proof:PropDeltaBounded}. As a result of the above proposition, we next prove that MATB-P and the relaxed policy approach each other as $N \rightarrow \infty$, which would then (as a result of \eqref{Cost_comparison}) imply that MATB-P is asymptotically optimal for Problem \ref{Prob:Cap_Cons}. To this end, we define an auxiliary policy $\hat{\gamma}$, %which transmits exactly the same agents as the 
under which the AoI sample paths are the same as those under the relaxed policy $\gamma_R$, but for each additional agent that is not supposed to be transmitted by MATB-P, it adds a penalty to the cost as: %existing cost as: 
\begin{align}\label{Age_Penalty}
    \omega(y,A,C_W) =&  c(\bar{\Delta},A,C_{W}) \times \mathbf{1}_{\{(1-\frac{\mathcal{C}}{n^\lambda_k}) >0\}}\mathbf{1}_{\{\tau \geq y\}}.
\end{align}
Further, we let $\{\tilde{\tau}^i_k\}_{k=1}^\infty$ and $\{\tau^i_k\}_{k=1}^\infty$ to be the sequences of AoIs of the $i^{th}$ agent under MATB-P and $\gamma^{i}_R$ (or equivalently $\hat{\gamma}^{i}$), respectively. Then, it is easy to see that $\omega(\tilde{\tau}^i(t),A_i,C_{W^i})$ dominates $c(\tilde{\tau}^i(t),A_i,C_{W^i})$, $\forall i,k$. 
As a consequence, it follows that
\begin{align}\label{Cost_comparison}
    J^{BS}({\gamma_R}) \leq J^{BS}({\gamma^*}) \leq J^{BS}({\gamma}) \leq J^{BS}({\hat{\gamma}}), 
\end{align}
where $\gamma^*$ is any optimal policy that solves Problem \ref{Prob:Cap_Cons}. Then, we have the following result.

% Define $\bar{a}:=\max_{\Phi}\|A(\phi_i)\|_F$, $C_W:=\max_{\Phi}C_{W}(\phi_i)$, $\kappa:= \max_{\Phi}\kappa(\phi_i)$, $m(\bar{a}):= \bar{a}^{2\kappa+2}(\bar{a}^2-1)^2$ and let $\operatorname{U} > 0$ be a constant.

% Then, we have the following assumption.

% \begin{assumption}\label{As_For_k}
% It holds that
%     \begin{align*}
% & m(\Bar{a})(\bar{\Delta}_N + 2\kappa + (2-\Bar{a}^2)\Bar{a}^{2\kappa+2}) \Bar{a}^{2\Bar{\Delta}_N} = \operatorname{U} \frac{(\Bar{a}^2-1)^4}{tr(C_W)} N^{1/4}
%     % & \log \Bar{a}^2  \frac{\log \Bar{a}^2}{-\log p} \log \left( m(a)\left(\frac{\log 2N/\delta}{-\log p} \right) + 2\kappa + (2-\Bar{a}^2)\Bar{a}^{2\kappa+2}\right) \nonumber\\
%     % & = \log N^{1/4} + \log \left( \operatorname{U}(\bar{a},p)(\Bar{a}^2-1)^4 / tr(C_W)\right),
% \end{align*}
% \end{assumption}

% \begin{assumption}\label{AsAbound}
% The inequality $0 < \|A(\theta)\|_F < \sqrt{1/(1-\alpha)}$ holds $\forall \theta \in \Theta$, where $\alpha = \frac{R_d}{N}$.
% \end{assumption}
% We remark that the Assumption \ref{AsAbound} entails finite scheduling cost under the hard-bandwidth policy, and further insights on it are presented later (in Remark \ref{remark_AboundedAssump}).

\begin{theorem}\label{Th:Asymptotic_optimality_deterministic}
Let $\alpha$ be fixed and suppose that Assumption \ref{As:sys_param} holds. Then, the difference in the scheduling cost under MATB-P and $\gamma_R$ converges to 0 exponentially fast as a function of $N$. Consequently, as $N \rightarrow \infty$, MATB-P becomes asymptotically optimal for Problem \ref{Prob:Cap_Cons}.
\end{theorem}
The proof of Theorem~\ref{Th:Asymptotic_optimality_deterministic} is provided in Appendix \ref{Ap:asymp_opt}. Next, we provide a remark on Proposition \ref{Prop:Delta_boundedness} and Theorem \ref{Th:Asymptotic_optimality_deterministic}.

% \begin{remark}
% We remark here that \eqref{Parameter_dependent_bound} depends on the combination of instability of the agents and the erasure probability. This reinstates the fact that under channel erasure, we cannot stabilize arbitrarily unstable agents, which has also been seen in different contexts in the literature \cite{imer2006optimal,moon2014discrete}.
% \end{remark}
\begin{remark}
    % We make a couple of remarks here. 
    As a consequence of Proposition \ref{Prop:Delta_boundedness}, in the case of deterministic channel, no assumptions are needed on the system parameters to prove the asymptotic optimality of the MATB protocol. This is thus a significant relaxation of the result given in \cite{aggarwal2022weighted}, where an upper bound on $\|A(\theta)\|_F$ was required. Second, we note that Theorem \ref{Th:Asymptotic_optimality_deterministic} proposes an exponential order of convergence of MATB-P toward optimality as $N \rightarrow \infty$, which is sharper than the $\mathcal{O}(1/\sqrt{N})$ convergence bound obtained in \cite{aggarwal2022weighted}.
\end{remark}
Next, we proceed to the general case with $p >0$. The following theorem shows that under MATB-P, the AoI takes large values with arbitrarily small probability.

\begin{theorem}\label{Prop:Delta_High_deviation_bounded}
    Let $\alpha$ be fixed and $p > 0$. Then, given $\delta \in (0,1)$, the upper confidence bound on AoI $\tau^i_k =  \Os(\log(1/\delta)),$ $ \forall i$, $ \forall k$ with probability at least $1-\delta$. %\strike{where $\tilde{\Delta}(\delta)$ is the confidence bound on AoI depending on $\delta$}.
\end{theorem}
The proof of the theorem can be found in Appendix~\ref{appen_prop3}. Further, it shows that the AoI has a vanishing tail under MATB-P, which can be used to give high-probability guarantees on the freshness of information under high traffic. Next, to prove asymptotic optimality of MATB-P under the case with erasure-prone channel, we again consider an auxiliary policy $\Check{\gamma}$, which transmits agents according to $\gamma_R$, except that, for each agent which is not supposed to be transmitted by MATB-P, it adds an additional penalty to the cost, which (by a slight abuse of notation) is defined as:

\begin{small}
\begin{align*}%\label{Age_Penalty_with_erasure}
    \omega(y,A,C_W) =&  \sum_{\ell=1}^{\infty}p^\ell c(\tau+\ell,A,C_{W}) \times \mathbf{1}_{\{(1-\frac{\mathcal{C}}{n^\lambda_k}) >0\}}\mathbf{1}_{\{\tau \geq y\}},
\end{align*}
\end{small}such that $\omega(\tilde{\tau}^i_k,A_i,C_{W^i})$ dominates the expected WAoI $c(\tilde{\tau}^i_k,A_i,C_{W^i})$, for all $i$, $k$. Further, $\omega(y,A,C_W) < \infty$ as a consequence of Assumption \ref{As:sys_param}. Thus, using similar arguments as for Theorem \ref{Th:Asymptotic_optimality_deterministic}, we can prove the following main result.

\begin{theorem}\label{Th:erasure_system}
Let $\alpha$ be fixed and $0<p<1$. Then, MATB-P approaches the optimal policy $\gamma^*$ for Problem \ref{Prob:Cap_Cons} exponentially fast as the number of agents grows.
\end{theorem}

\begin{remark}
    We note here that in most literature such as \cite{ayan2019age,Ayan2020,hatami2022demand}, the authors rely on truncating the AoI state space to a sufficiently large value and consequently working with a finite space to derive the corresponding scheduling policies. Here, however, we do not require any such truncation on the state space. This is more natural since in communication systems with non-zero erasure probability, the AoI can always exceed the truncation value, even if the probability of the same tends to 0.  
\end{remark}
The solution to the original capacity-constrained problem is thus completely characterized, and we next proceed to solving the finite-agent game problem.

\section{Solution to Agent-Level Game Problem}
\label{sec:MFG_summ}
In this section, we solve the agent-level game problem by using the BS's scheduling policy as constructed in the previous section. Typically, networked problems involve a large number of users, and thus, belong to the class of \textit{large population} games. Characterizing Nash equilibria based on a centralized information structure introduced in Section \ref{sec:formulation} in such a setting is therefore unrealistic. Thus, the objective here is to characterize decentralized Nash policies for each agent. For that purpose, we first consider a limiting game (or the MFG) with a countably infinite number of players. Then, we characterize the equilibrium of the MFG (called the MFE) by utilizing the Nash certainty equivalence principle \cite{huang2007large}. As a result of the latter, each agent's effect on the aggregate behavior becomes negligible, which gives rise to the notion of a representative agent solving a decentralized  
 stochastic optimal control problem using only local information, by playing against the aggregate distribution. Consequently, we also show that the MFG solution provides an approximate Nash solution for the finite-agent game.

 \subsection{Decentralized Stochastic Optimal Tracking Problem}
Consider a generic agent of type $\phi$ from the infinite population, whose plant dynamics evolve as

\begin{align}\label{generic_dyn}
    X_{k+1} = A(\phi_i)X_k + B(\phi_i)U_k + W_k, ~k \geq 0,
\end{align}
 where $X_k \in \mathbb{R}^n$ and $U_k \in \mathbb{R}^m$ denote the state and control input of the generic agent, respectively. $W_k \in \mathbb{R}^n$ is an i.i.d zero mean Gaussian noise with positive definite covariance $C_W(\phi)$. The initial state $X_0$ has symmetric density with mean $x_{\phi,0}$ and covariance $\Sigma_x > 0$. The decoder and the controller information structures are same as in subsection \ref{sec:formulation}, except with the superscript $i$ removed. The objective of the controller is to minimize the function
 \begin{align}\label{MFG:cost}
     \!\!J(\xi,\mu)\!:=\! \overline{\lim\limits_{T \rightarrow \infty}} \frac{1}{T}\mathbb{E}\left[\sum_{k=0}^{T-1} \|X_k - \mu_k\|^2_{Q(\phi)}\! +\! \|U_k\|^2_{R(\phi)}  \right]\!\!,\!\!
 \end{align}
where the policy
% $\xi \in \Xi:= \{\xi \mid \xi \text{ is adapted to } \sigma(I^{d,con}_\ell), \ell = 0, \cdots, k \}$ and
$\xi \in \Xi$ is adapted to the decentralized information structure
$I^{d,con}_0:= Z_0,~I^{d,con}_k:= \{U_{0:k-1}, Z_{0:k}\}, \forall k \geq 1$ being the decentralized information structure of the generic agent. Note that this is different from the centralized information structure, which involves the information of all the other agents as well. Further, $\mu=(\mu_k)_{k \geq 0} \in \mathcal{M}:= \{\mu_k \in \mathbb{R}^n \mid \|\mu\|_{\infty}:= \sup_{k \geq 0} \|\mu_k\| \leq \infty\}$, also called the MF trajectory, denotes the infinite agent approximation to the consensus term $(\mu^N_k)_{k \geq 0}$ in \eqref{Finite_cost}. This term leads to decoupling between the otherwise cost-coupled agents in the finite-agent game, and the resulting problem becomes a linear-quadratic tracking (LQT) problem, for which the optimal policy is well known (and is provided in Proposition \ref{Prop:Optimal_Cntrl}).

Next, we introduce the operator $\Psi:\mathcal{M} \rightarrow \Xi$ which defines the mapping $\mu \mapsto \xi$ and the operator $\Theta: \Xi \rightarrow \mathcal{M}$, which defines the mapping $\xi \mapsto \mu$. While the former generates an optimal policy given a MF trajectory $\mu$, the latter computes a trajectory from a given control policy. 

The MFE can then be defined as the pair $(\xi^*,\mu^*)$ such that $\mu^*$ is the fixed point of the composite operator $\Theta \circ \Psi$, i.e., $\mu^* = \Theta \circ \Psi (\mu^*)$.
Now, we state the following proposition which characterizes the optimal control policy for the LQT problem of the generic agent.

\begin{proposition}\label{Prop:Optimal_Cntrl}
    Suppose that the Assumption \ref{As:sys_param} holds and consider the dynamics \eqref{generic_dyn} with cost \eqref{MFG:cost}. Then, the following are true:
    \begin{enumerate}
        \item The optimal control action of the generic agent is:
        \begin{align}\label{Optimal_Cntrl}
            U^*_k = -K_1(\phi)Z_k - K_2(\phi)g_{k+1}
        \end{align}
        where $K_2(\phi) =   (R(\phi) +B(\phi)^\top K_1(\phi)B(\phi))^{-1}B(\phi)^\top$, 
$K_1(\phi) = K_2(\phi)K(\phi)A(\phi)$,
and $K(\phi) > 0$ is the unique solution to
\begin{align*}
K(\phi) \!= \!A(\phi)^\top \![K(\phi)A(\phi) \!-\! K(\phi)^\top\! B(\phi)K_1 (\phi)] \!+\! Q(\phi).
\end{align*}
Further, the trajectory $g_k$ satisfies the backward dynamics $g_k = A_{cl}(\phi)^\top g_{k+1} Q(\phi)\mu_k$, with the initial condition $g_0 = -\sum_{j=0}^{\infty}{(A_{cl}(\phi)^j})^\top$ $Q(\phi)\mu_j$ and $A_{cl}(\phi) = A(\phi) - B(\phi)K_1(\phi)$ being Hurwitz. In addition, the dynamics for $g_k$ has a unique solution in $\mathcal{M}$, which can be given as $g_k = -\sum_{j=k}^{\infty}{(A_{cl}(\phi)^{j-k})}^\top Q(\phi)\mu_j$.

\item The optimal cost is bounded above as:
\begin{align}\label{LQGCost}
    & J(\xi,\mu^*) \!\leq \!tr(K(\phi)C_W(\phi))\! +\! \overline{\lim\limits_{T \rightarrow \infty}} \frac{1}{T}\!\sum_{k=0}^{T-1}\!\!\xi_k^\top \!Q(\phi) \xi_k \nonumber \\ &  -  g_{k+1}^\top B(\phi)K_2(\phi)g_{k+1} \!+ \!\|A(\phi)^\top\!\! K(\phi)^\top\!\! B(\phi)K_1(\phi)\| \nonumber \\ & \times \left( \!\sum_{m=1}^{\hat{\kappa}}\sum_{r=1}^{m}tr(A(\phi)^{{r-1}^\top}\!\!A(\phi)^{r-1}C_W(\phi))\! \right. \nonumber \\
    & \left. +\! \frac{\|C_W(\phi)\|_F}{\|A(\phi)\|_F^2\!-\!1}\! \times\! \left[\frac{\|A(\phi)\|^{2\hat{\kappa} + 2}_F p}{1-\|A(\phi)\|^2_Fp} - \frac{p}{1-p}\right] \right).
    \end{align}
    \end{enumerate}
\end{proposition}
\begin{proof}
    The proof of 1) follows from \cite{aggarwal2022weighted}. For the proof of 2), we substitute \eqref{Optimal_Cntrl} in \eqref{MFG:cost}, to arrive at
\begin{align}\label{temp_costLQ0}
    & J(\xi,\mu^*) \!\leq \!tr(K(\phi)C_W(\phi))\! +\! \overline{\lim\limits_{T \rightarrow \infty}} \frac{1}{T}\!\sum_{k=0}^{T-1}\!\!\xi_k^\top \!Q(\phi) \xi_k \nonumber \\ &  -  g_{k+1}^\top B(\phi)K_2(\phi)g_{k+1}  \nonumber \\ & + \overline{\lim\limits_{T \rightarrow \infty}} \frac{1}{T}\sum_{k=0}^{T-1}\!\|A(\phi)^\top\!\! K(\phi)^\top\!\! B(\phi)K_1(\phi)\|\mathbb{E}\left[\|e_k\|^2 \right].
\end{align}
% Further, we have
% \begin{align*}
%     & \mathbb{E}\left[\|e_k\|^2 \right] \leq   \sum_{r=1}^{\overline{\kappa}}tr(A(\phi)^{{r-1}^\top}A(\phi)^{r-1}C_W(\phi)) \nonumber \\
%     & + \sum_{r=\overline{\kappa}+1}^k tr(A(\phi)^{{r-1}^\top}A(\phi)^{r-1}C_W(\phi))p^{k-\Bar{\kappa}-1},
%     % & \leq T\sum_{r=1}^{\overline{\kappa}^{\phi}(\overline{\lambda}^*)}tr(A(\theta)^{{r-1}^\top}A(\theta)^{r-1}K_W(\theta)) \nonumber \\ & + \sum_{k=\overline{\kappa}^{\phi}(\overline{\lambda}^*)+1}^{T-1} \|K_W(\theta)\|_F\frac{(1-\|A(\theta)\|_F^{2(k-\tau_u-1)})(1-\alpha)^{k-\overline{\kappa}^{\phi}(\overline{\lambda}^*)-1}}{1-\|A(\theta)\|_F^2} \nonumber \\
% \end{align*}
% where $\bar{\kappa}:= \overline{\kappa}^{\phi}(\overline{\lambda}^*)$ and the inequality follows using the hard-bandwidth policy of Section \ref{subsec:hard_cons_pol}. Then, using Assumptions \ref{As:sys_param}-\ref{As_For_k}, summing over 0 to $T-1$, dividing by T, taking limsup, and substituting in \eqref{temp_costLQ0}, we get \eqref{LQGCost}.

Consider the following:
\begin{align}\label{temp_costLQ}
& \mathbb{E}\left[\|e_k\|^2 \right] \!\!=\!\! \sum_{m=1}^\infty \sum_{r=1}^{m}tr(A(\phi)^{{r-1}^\top}\!\!A(\phi)^{r-1}C_W(\phi)) \mathbb{P}(\tau_k = m) \nonumber \\
    & \leq   \sum_{m=1}^{\hat{\kappa}}\!\sum_{r=1}^{m}tr(A(\phi)^{{r-1}^\top}\!\!A(\phi)^{r-1}C_W(\phi)) \nonumber \\
    & +\!  \!\sum_{m=\hat{\kappa}+1}^{\infty} \sum_{r=1}^m \!tr(A(\phi)^{{r-1}^\top}\!\!A(\phi)^{r-1}C_W(\phi))p^{m-\tau_u} \nonumber \\
    % & \leq T\sum_{r=0}^{\tau_u}tr(A(\theta)^{{r-1}^\top}A(\theta)^{r-1}K_W(\theta)) \nonumber \\ & + \sum_{k=\tau_u+1}^{T-1} \|K_W(\theta)\|_F\frac{(1-\|A(\theta)\|_F^{2(k-\tau_u-1)})(1-\alpha)^{k-\tau_u-1}}{1-\|A(\theta)\|_F^2} \nonumber \\
    & \!\! \leq \!\sum_{m=1}^{\hat{\kappa}}\sum_{r=1}^{m}tr(A(\phi)^{{r-1}^\top}\!\!A(\phi)^{r-1}C_W(\phi))\! \nonumber \\
    & +\! \frac{\|C_W(\phi)\|_F}{\|A(\phi)\|_F^2\!-\!1}\! \times\! \left[\frac{\|A(\phi)\|^{2\hat{\kappa} + 2}_F p}{1-\|A(\phi)\|^2_Fp} - \frac{p}{1-p}\right],
\end{align}
where $\hat{\kappa}:= \overline{\kappa}^{\phi}(\overline{\lambda}^*)$ and the last inequality follows using the scheduling policy of Section \ref{subsec:hard_cons_pol}, Assumption \ref{As:sys_param}, and the fact that $\|AB\|_F \leq \|A\|_F\|B\|_F$. Then, combining \eqref{temp_costLQ0} and \eqref{temp_costLQ}, we arrive at \eqref{LQGCost}. This completes the proof.% Then, using Assumptions \ref{As:sys_param}-\ref{As_For_k}, summing over 0 to $T-1$, dividing by T, taking limsup, and substituting in \eqref{temp_costLQ0}, we get \eqref{LQGCost}.
\end{proof}
\begin{remark}
    We remark here that the boundedness of the cost in the special case of deterministic channels is similarly implied by the uniform bound on the AoI from Proposition \ref{Prop:Delta_boundedness}. This reiterates the advantage of the MATB-P over the uniformly randomized policy in \cite{aggarwal2022weighted}, where an assumption on $\|A\|_F$ was required to entail the boundedness of the cost.
\end{remark}

\subsection{$\epsilon$--Nash Equilibrium}
Now, that we have computed the optimal policy of the generic agent of type $\phi$, we will henceforth prove the existence of a unique MFE. To this end, we use the policy $\gamma$ from Section \ref{subsec:hard_cons_pol} to arrive at the closed-loop system (CLS) in \eqref{Decoder_state} under the policy \eqref{Optimal_Cntrl} as
\begin{align}
    Z_{k+1} &= (A_{cl}(\phi)Z_k -B(\phi)K_2(\phi)g_{k+1} + W_{k+1}) \mathbf{1}_{[\wp_{k+1}=1]} \nonumber \\
    & + (A_{cl}(\phi)Z_k -B(\phi)K_2(\phi)g_{k+1}) \mathbf{1}_{[\wp_{k+1}=0]},
\end{align}
which on taking expectation and using Proposition \ref{Prop:Optimal_Cntrl} yields
\begin{align}
    \mu^\phi_k :=& \mathbb{E}[X_k]:= A_{cl}(\phi)^kx_{\phi,0} + \sum_{j=0}^{k-1} A_{cl}(\phi)^{k-j-1}B(\phi)K_2(\phi) \nonumber \\
     &\hspace{1cm} \times \sum_{r=j+1}^\infty (A_{cl}(\phi)^{r-j-1})^\top Q(\phi) \mu_r.
\end{align}
Define the MF operator as
\begin{align}
    \{\operatorname{M}_{\operatorname{F}}(\mu)\}_k:= \sum_{\phi \in \Phi} \mu^\phi_k \mathbb{P}(\phi), k \geq 0.
\end{align}
Also, we invoke the following assumption on model parameters.

\begin{assumption}\label{As:Contraction}
    $\|A_{cl}(\phi)\| + \sum_{\phi \in \Phi}\|Q(\phi)\| \|B(\phi)K_2(\phi)\|(1\!-\!\|A_{cl}(\phi)\|)^{-2}\mathbb{P}(\phi) < 1$, $\forall \phi \in \Phi$.
\end{assumption}
We next prove the following lemma and state the main theorem showing the $\epsilon$--Nash property of the MFG solution.
\begin{lemma}\label{Combined_results}
    Suppose that Assumptions \ref{As:sys_param}-\ref{As:Contraction} hold. Then, the following are true:
    \begin{enumerate}
        \item (MFE Uniqueness): There exists a unique $\mu^*$ such that $\mu^*= \operatorname{M}_{\operatorname{F}}(\mu^*)$ with the property that $\exists K_3^* \in \mathbb{K}:= \{K_3 \in \mathbb{R}^{n \times n} \mid \|K_3\| < 1, \mu^*_{k+1} = K_3 \mu^*_k\}$, and $\mu^*_0 = \sum_{\phi \in \Phi}x_{\phi,0}\mathbb{P}(\phi)$.
        
        \item (CLS stability): The CLS \eqref{system} under \eqref{Optimal_Cntrl} is mean-squared stable, i.e., $\sup_{N \geq 1}\max_{1 \leq j \leq N}$ $\overline{\lim}_{T \rightarrow \infty} $ $\frac{1}{T} \sum_{k=0}^{T-1}\mathbb{E}[\|\mu^*_k\|^2] <$ $ \infty$.

        \item (MFE Approximation): We have that $\mu^{N,*}_k \xrightarrow[\text{$N \rightarrow \infty$}]{\text{m.s.}} \mu^*_k$ at a rate of $\mathcal{O}(1/ \min_{\phi}N_{\phi})$, where $N_{\phi}$ denotes the cardinality of agents of type $\phi$, and $\mu^{N,*}_k$ is the empirical state average under \eqref{Optimal_Cntrl}.
    \end{enumerate}
\end{lemma}
\begin{proof}
    The proof of parts 1) and 3) follow in a similar manner as Theorem 3 and Proposition 6 in \cite{aggarwal2022weighted}. For part 2), consider the following, with the superscript $*$ dropped for ease of notation. Substituting \eqref{Optimal_Cntrl} in \eqref{generic_dyn}, we arrive at the closed-loop system as
\begin{align}\label{clsys}
    X^i_{k+1} =& A_{cl}(\phi_i) X^i_k + B(\phi_i)K_1(\phi_i) e^i_k \nonumber \\ &- B(\phi_i)K_2(\phi_i)g^i_{k+1} + W^i_k.
\end{align}
% Equivalently, \eqref{clsys} can be written as 
% \begin{align}\label{CL_system}
%     X^i[k] =& H(\theta_i)^k X^i[0] + \sum_{p=0}^{k-1}H(\theta_i)^{k-p-1}B(\theta_i)\Pi(\theta_i) e^i[p] \nonumber \\ &- \sum_{p=0}^{k-1}H(\theta_i)^{k-p-1}B(\theta_i)L(\theta_i) r^i[p+1] \nonumber \\ & + \sum_{p=0}^{k-1}H(\theta_i)^{k-p-1}W^i[p].
% \end{align}
Then, from \eqref{clsys} we have that
\begin{align}\label{Expected_CL_system}
     &\mathbb{E}\left[\|X^i_k\|^2 \right] \leq 4\mathbb{E} \left[ \left\lVert A_{cl}(\phi_i)^k X^i_0 \right\rVert^2 \right]  \nonumber \\
    & +4\mathbb{E} \left[ \left\lVert \sum_{r=0}^{k-1}A_{cl}(\phi_i)^{k-r-1}B(\phi_i)K_1(\phi_i) e^i_r \right\rVert^2 \right] \nonumber \\
    & +4\mathbb{E} \left[ \left\lVert \sum_{r=0}^{k-1}A_{cl}(\phi_i)^{k-r-1}B(\phi_i)K_2(\phi_i) g^i_{r+1} \right\rVert^2 \right]  \nonumber \\
     &+4\mathbb{E} \left[ \left\lVert \sum_{r=0}^{k-1}A_{cl}(\phi_i)^{k-r-1}W^i_r \right\rVert^2 \right],
\end{align}
where we used the fact that $\|\sum_{i=1}^k x_i\|^2 \leq k \sum_{i=1}^k\|x_i\|^2$.
We note that since $A_{cl}(\phi)$ are Hurwitz (as a result of proposition \ref{Prop:Optimal_Cntrl}), using \cite[Theorem 3.9]{costa2006discrete}, we can bound the first term in \eqref{Expected_CL_system} by $\iota(\phi_i)tr(\Sigma_x + x_{\phi_i,0}x_{\phi_i,0}^\top)/(1-\varsigma(\phi_i))$, and the fourth term in \eqref{Expected_CL_system} by $\iota(\phi_i)\sup_{1 \leq j \leq p} C_{W}(\phi_i)/(1-\varsigma(\phi_i))$ for constants $\iota(\phi_i) \geq 1$ and $0<\varsigma(\phi_i) <1$. Similarly, using the fact that $\|g^i\|_\infty < \infty$ (from Proposition \ref{Prop:Optimal_Cntrl}), the third term in \eqref{Expected_CL_system} can be bounded by $\frac{2\iota(\phi_i)\|B(\phi_i)K_2(\phi_i)\|^2\|g^i\|_\infty^2}{(1-\sqrt{\varsigma(\phi_i)}) (1-\varsigma(\phi_i))}$. Finally, using similar arguments as for the third term, we can show that the second term (call it $T_2$) can  be bounded as
\begin{align*}
    & T_2 \leq \sum_{r=0}^{k-1} \iota(\phi_i)\varsigma(\phi_i)^{k-r-1}\|B(\phi_i)K_1(\phi_i)\|^2 \mathbb{E}\left[\|e_r\|^2 \right] \nonumber \\ 
    & + \!\!\! \sum_{\substack{r,s=0,\\ r \ne s}}^{k-1} \!\!\!\iota(\phi_i)\varsigma(\phi_i)^{k\!-\frac{r}{2}\!-\!1\!-\frac{s}{2}}\|B(\phi_i)K_1(\phi_i)\|^2\!\! \sqrt{\!\mathbb{E}\left[\|e_r\|^2 \right]\!\mathbb{E}\left[\|e_s\|^2 \right]} \nonumber \\
    & \leq \beta(\phi_i)\iota(\phi_i)\|B(\phi_i)K_1(\phi_i)\|^2\!\left[\!\frac{1}{1-\varsigma(\phi_i)}\!+\! \frac{ 1}{(1\!-\!\sqrt{\varsigma(\phi_i)})^2}\right]\!\!,
\end{align*}
where $\beta(\phi_i) \!=\! \!\sum_{m=1}^{\Bar{\kappa}}\!\sum_{r=1}^{m}tr(A(\phi_i)^{{r-1}^\top}\!\!A(\phi_i)^{r-1}C_W(\phi_i))\!+\! \frac{\|C_W(\phi_i)\|_F}{\|A(\phi_i)\|_F^2\!-\!1}\! \times\! [\frac{\|A(\phi_i)\|^{2\Bar{\kappa}^{\phi_i} + 2}_Fp}{1-\|A(\phi_i)\|^2_Fp} \!-\! \frac{p}{1-p}]$. Finally, summing up all the bounds and noting that $\Phi$ is a finite set, we have the desired result. The proof is thus complete.
\end{proof}

We next state the following definition below.

\begin{definition}[Approximate Nash equilibrium]\label{epsNashDefn}
Given $\epsilon>0$, the set of control policies $\{\xi^j\}_{j \in [N]}$ constitutes an $\epsilon$--Nash equilibrium for the cost functions $\{J_j\}_{i \in [N]}$, if 
\begin{align}\label{eps_delta_Nash}
J_i(\xi^{*,i}, \xi^{*,-i}) \leq \inf_{\pi_c^i \in \Pi_i} J_i(\pi_c^{i}, \mu^{*,-i}) + \epsilon,~~\forall i \in [N].
\end{align}
\end{definition}

Then, we present the main result of this section stating that the %control laws extracted from the MFG formulation 
MFE control laws constitute an $\epsilon$-Nash equilibrium for the finite-population case.

\begin{theorem}\label{Th:eps_Nash}
    Suppose that Assumptions \ref{As:sys_param}-\ref{As:Contraction} hold. Then the sequence of decentralized control policies $\{\xi^j\}_{j \in [N]}$, constitutes an $\epsilon$--Nash equilibrium for the $N$--agent capacity-constrained LQ-mean field game. In particular, we have that
\begin{align}\label{epsNash}
\!\!J_i(\xi^{*,i}\!,\! \xi^{*,-i})\! \leq \!\inf_{\pi_c^i \in \Pi_i}\!\! J_i(\pi_c^{i}, \mu^{*,-i}) \!+ \!\Os\!\! \left(\! \frac{1}{ \sqrt{\min_{\phi \in \Phi} \!N_\phi\!}}\!\right)\!.\!\!
\end{align}
\end{theorem}

\begin{proof}
    The proof follows from Lemma \ref{Combined_results} using techniques similar to those in \cite{aggarwal2022linear,moon2014discrete}, and hence is omitted.
\end{proof}

\section{An Illustrative Example}
\label{sec:example}
In this section, we validate the theoretical results using a numerical example. We first demonstrate the asymptotic optimality of MATB-P. For this purpose, we consider values for $N$ from 5 till 100, a time horizon of 5000 seconds, a low capacity $\mathcal{C} = 0.25N$, and an erasure probability $p = 0.2$. We plot the average weighted AoI of the system as a function of $N$ in Fig. \ref{fig:scheduling_policy}, for both the relaxed policy and MATB-P. We can see that the difference in the average cost under the above decays to 0, which shows the asymptotic optimality of the MATB policy.
%corroborates the results of Theorem \ref{Th:erasure_system}.

\begin{figure}[!h]
     \centering
    %  \begin{subfigure}[b]{0.21\textwidth}
    %      \centering
         \includegraphics[width=0.44\textwidth]{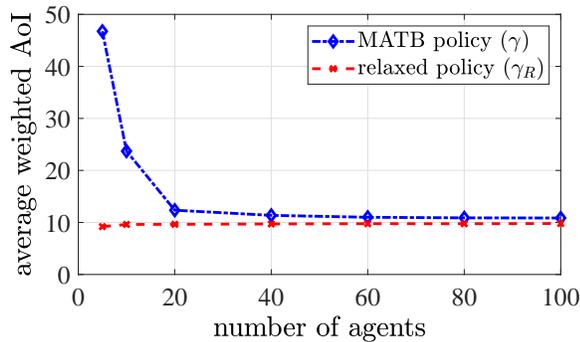}
        %  \caption{Hard-bandwidth policy approximates relaxed policy as $N \rightarrow \infty$.}
        %  \label{fig:schedul_short}
    %  \end{subfigure}
    %  \hfill
    %  \begin{subfigure}[b]{0.23\textwidth}
        %  \centering
    %      \includegraphics[width=\textwidth]{figs/Avg_cost_Vs_Bandwidth_ratio.eps}
    %      \caption{Average cost per agent for hard-bandwidth policy as $\alpha$ increases.}
    %      \label{fig:consensus_short}
    %  \end{subfigure}
    \vspace{-0.3cm}
        \caption{Plot shows the performance of the relaxed policy ($\gamma_R$) and the MATB policy, converging to each other.}
        \label{fig:scheduling_policy}
        \vspace{-0.5cm}
\end{figure}

Next, we simulate the behavior of a $900$--agent system, with 3 types of (scalar) agents, namely, with $A=0.5,1.0,1.15$, under the MATB scheduling protocol and the MFE policy $\xi^*$. We take $B=0.1269, C_W = 5$, $Q=R=2$, and a horizon of 500 seconds. In Fig. \ref{fig:sched_cons_short}, in the left plot, we show the variation of the average cost per agent as a function of the available capacity for a fixed erasure probability $p=0.2$. Next, in the right, we show the variation of the average cost per agent as a function of the channel erasure for a fixed capacity ratio $\alpha = 0.45$. The figures show a box plot depicting the median (red line) and spread (box) of the average cost per agent over 100 runs for each value of $\alpha$, and $p$, respectively. We can easily see that the average cost varies inversely with the available downlink capacity and in direct proportion to the erasure probability, aligned with intuition. 

\begin{figure}[!h]
     \centering
    %  \begin{subfigure}[b]{0.21\textwidth}
    %      \centering
         \vspace{-0.35cm}
         \includegraphics[width=0.495\textwidth]{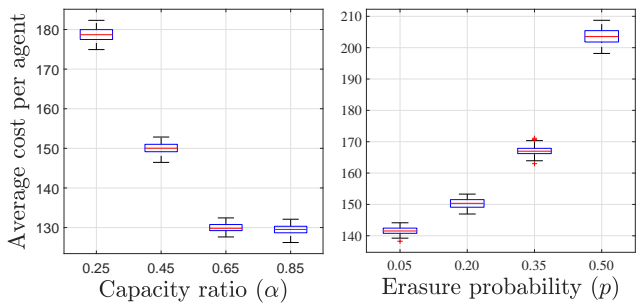}
        %  \caption{Hard-bandwidth policy approximates relaxed policy as $N \rightarrow \infty$.}
        %  \label{fig:schedul_short}
    %  \end{subfigure}
    %  \hfill
    %  \begin{subfigure}[b]{0.23\textwidth}
        %  \centering
    %      \includegraphics[width=\textwidth]{figs/Avg_cost_Vs_Bandwidth_ratio.eps}
    %      \caption{Average cost per agent for hard-bandwidth policy as $\alpha$ increases.}
    %      \label{fig:consensus_short}
    %  \end{subfigure}
    \vspace{-0.5cm}
        \caption{Plots show the variation of aggregate cost per agent with (a) capacity ratio $\alpha$, and (b) erasure probability $p$.}
        \label{fig:sched_cons_short}
        \vspace{-0.5cm}
\end{figure}

\section{Conclusion}\label{sec:conclusion}
In this paper, we have formulated a large population game problem involving information transmission over unreliable networks, thereby extending the setting of \cite{aggarwal2022weighted} and improving the guarantees in the special case of \cite{aggarwal2022weighted}. The network is regulated by a BS, for which we have constructed an asymptotically optimal scheduling policy. We have provided a tail analysis of the AoI under the same, first, for the case when the channel is free of any erasure and then for the case with erasure. Next, by using this policy, we have solved the consensus problem between the non-cooperative agents using the MFG framework by proving the existence of a unique equilibrium and consequently showing its $\epsilon$--Nash property to the finite-agent game problem. Finally, we have simulated a numerical example, which corroborates the theoretical developments.

\bibliographystyle{IEEEtran}
\bibliography{references}

% \appendix
\begin{appendices}
\section{Computation of $\kappa$}\label{ap:Compute_kappa}
% \subsection{Computation of $\kappa$}
% \begin{proof}[Computation of $\kappa$]
Consider the finite state-space $\operatorname{S}' = \{0, 1, \cdots, \kappa\}$. Then, the Bellman equation for Problem \ref{decoupled_problem} %\eqref{Prob:MDP} %for the cost function $\operatorname{C}(\cdot, \cdot)$ 
can be written as
\begin{align}\label{Bellman_eqn}
    V(\tau) + \sigma^* &= \min\{\operatorname{C}(\tau,0) + V(\tau+1), \operatorname{C}(\tau,1) \nonumber \\
    & + pV(\tau+1) + (1-p)V(0)\}.
\end{align}
Then, by using \eqref{Bellman_eqn}, for $\tau \geq \kappa$, we have that
\begin{align*}
    &V(\tau) = c(\tau)+\lambda + pV(\tau+1) + (1-p)V(0) - \sigma^* \nonumber \\
    &\! = \!\sum_{r=0}^{n-1}(c(\tau+r) \!+\! \lambda + (1-p)V(0) \!-\!\sigma^*)p^r \!+\! p^nV(\tau+n).
\end{align*}
Taking the limit as $n \rightarrow \infty$, we get
\begin{align}\label{temp:feqn}
   \! V(\tau)\! =\!\frac{\lambda \!+ \!(1\!-\!p)V(0)\! -\! \sigma^*}{1\!-\!p} \!+ \!\underbrace{\lim_{n \rightarrow \infty} \sum_{r=0}^{n-1}c(\tau+r)p^r}_{f(\tau)},
\end{align}
where we define $c(x):=c(x,A,C_{W})  = \sum_{r=1}^{x}tr((A^{r-1})^\top $ $A^{r-1}C_W)x=\!\sum_{r=1}^{x}tr\!\left(\!\left(A^{r-1}\!xC_W^{0.5}\right)^\top\! \!\!A^{r-1}\!C_W^{0.5}\right)\!x\! = \!\sum_{r=1}^{x} \|A^{r-1}C_W^{0.5}\|_F^2 x$.

Next, in order to give a closed-form equation to compute $\kappa$, we consider a scalar system (with $n=1$ in \eqref{system}) for the computation of the function $f(\cdot)$, which can be given as
\begin{align}\label{kappa_computation}
f(x) \!\!= \!\!\begin{cases}
\!\frac{C_W}{1-A^2}\!\left[ x(\frac{1}{1-p}\!-\!\frac{A^{2x}}{1-A^2p}) \!+ \!\frac{p}{(1-p^2)} \!-\! \frac{A^{2x+2}p}{(1-A^2p)^2} \right]\!\!, & \!\!\!\!\! A \neq 1, \\
\!\frac{C_Wx^2}{1-p} + \frac{2C_Wxp}{(1-p)^2} + \frac{C_Wp(1+p)}{(1-p)^3}, & \!\!\!\!\! A=1.
\end{cases}
\end{align}
We note that the calculation of $f(\cdot)$ involves an infinite sum which is finite under %as a result of 
Assumption \ref{As:sys_param}. Further, we observe that $V(\kappa) \leq \lambda/(1-p) + V(0) \leq V(\kappa+1)$. Hence, there exists $\eta \in [0,1]$ such that $V(\kappa+ \eta) = \lambda/(1-p) + V(0)$. Combining this with \eqref{temp:feqn}, we get that $\sigma^* = (1-p)f(\kappa+\eta)$. Next, for $\tau < \kappa$, we have from \eqref{Bellman_eqn} that $V(\tau) + \sigma^* = V(\tau+1) + c(\tau)$. Combining this with \eqref{temp:feqn}, evaluated at $\tau= \kappa$, we arrive at 
\begin{align}\label{kapp_computation1}
    (1+\kappa(1-p))f(\kappa+\eta) = \frac{\lambda}{(1-p)} \!+\! f(\kappa) \!+\! \sum_{i=1}^{\kappa-1}c(i).
\end{align}
The above is an implicit equation in $\kappa$ and $\eta$, for a given tuple $(A,C_W,p,\lambda)$, and can be solved in conjunction with \eqref{kapp_computation1} to compute $\kappa$.

\section{Proof of Proposition \ref{Prop:Delta_boundedness}}\label{Proof:PropDeltaBounded}
%[Proof of Proposition \ref{Prop:Delta_boundedness}]
Let us start by defining the sets $\bar{S}_k$ and $S^{\gamma}_k$ for $k \geq 0$:
\begin{align*}
    \bar{S}_k & := \{ i \in [N] | \tau^i_k > \bar{\kappa} := \max (\max_{\phi \in \Phi} \overline{\kappa}^{\phi}(\bar{\lambda}^*), \lceil \alpha^{-1} \rceil )\}, \\
    S^{\gamma}_k &  := \{ i \in [N] | \zeta^{i}_k = 1, t_0 \leq k \leq t_0 +\lceil \alpha^{-1} \rceil, t_0 \geq 0\}.
\end{align*}
% where $\overline{\kappa}(\phi):= \overline{\kappa}(A(\phi),C_W(\phi),\overline{\lambda}^*)$.
The set $\bar{S}_k$ is the set of agents whose AoIs exceed $\bar{\kappa}$, and hence the subset of agents which are supposed to be scheduled for transmission. The set $S^{\gamma}_k$ is the subset of agents which are scheduled at time $k$ using MATB-P. Let us also define the quantity $\tau_{\max,k} $ as:
\begin{align*}
    \tau_{\max,k} := \left\{\begin{array}{cc} \max\{\tau^i_k | i\in \bar{S}_k \}, & \text{if } |\bar{S}_k| > 0, \\
    \bar{\kappa}, & \text{if } |\bar{S}_k| = 0. \end{array} \right.
\end{align*}
This quantity is the upper bound on $\tau^i_k$ by definition. Now, we first investigate some properties of the sets $\bar{S}_k$ and $S^{\gamma}_k$. By definition, it holds that $|\bar{S}_{0}| = 0$. Define $t_0$ as the first timestep when the cardinality of $\bar{S}_k$ exceeds $\mathcal{C}$, i.e.,
\begin{align*}
    t_0 := \min_{k \geq 0} \{ k \mid \lvert\bar{S}_k\rvert > \mathcal{C}\}
\end{align*}
If $t_0 = \infty$, then $\tau^i_k$ can be trivially bounded by $\bar{\kappa}$. Hence, we assume $t_0 < \infty$. Let $m \geq 0$ be the timesteps it takes for the cardinality of $\bar{S}_k$ to drop below $\mathcal{C}+1$. More formally, we have that
\begin{align*}
    |\bar{S}_k| > \mathcal{C}, t_0 \leq k \leq t_0+m \text{ and } |\bar{S}_{t_0+m+1}| \leq \mathcal{C}.
\end{align*}
First, we notice that for any $k$ such that $t_0 \leq k \leq t_0+m$,
\begin{align}
    S^{\gamma}_k \subset \bar{S}_k.  \label{eq:SgammasubsetSbar}
\end{align}
This is due to the fact that for $k\in[t_0,t_0+m]$, we have that $|\bar{S}_k| > \mathcal{C}$, and thus, the number of agents to be scheduled is larger than $\mathcal{C}$. As a result, the scheduling policy breaks the tie using MATB-P and since the agents with the highest AoIs reside in $\bar{S}_k$ for $t_0\leq k \leq t_0+m$, they will be scheduled. This also means that
\begin{align}
    |S^{\gamma^d}_k| = \mathcal{C}, \text{ and } t_0 \leq k \leq t_0 + m. \label{eq:SgammaRD}
\end{align}
Next, we notice that if an agent with index $i \in \bar{S}_k \cap S^{\gamma}_k$ for $ t_0 \leq k \leq t_0 + m$, then
\begin{align}
i & \notin \bar{S}_{k'}, k < k' \leq t_0 + \min( m, \lceil \alpha^{-1} \rceil). \label{eq:inotinsbar}
\end{align}
This is due to the fact that if $i \in S^{\gamma}_k$ for any $i \in [N]$ and $k \geq 0$, then $\tau^i(k') \leq \lceil \alpha^{-1}\rceil$ for $k < k' \leq k+  \lceil \alpha^{-1}\rceil$, and thus, $i \notin \bar{S}_{k'}$, for $k < k' \leq k + \lceil \alpha^{-1}\rceil$. Further,  \eqref{eq:inotinsbar} follows since $t_0 + \min(m,\lceil \alpha^{-1} \rceil) \leq k+\lceil \alpha^{-1} \rceil$. Combining \eqref{eq:SgammasubsetSbar} and \eqref{eq:inotinsbar}, we can deduce that if $i \in S^{\gamma}_k, t_0 \leq k \leq t_0 + \min(m,\lceil \alpha^{-1} \rceil)$, then $i \notin S^{\gamma}(k'), k < k' \leq t_0 + \min(m,\lceil \alpha^{-1} \rceil)$. This means that if an agent is scheduled in the interval $[t_0,t_0+\min(m,\lceil \alpha^{-1}\rceil)]$, then it will not be scheduled again in this interval. Hence, 
\begin{align}\label{eq:Sgammanointersect}
     \!\!\!S^{\gamma}_k \cap S^{\gamma}_{k'} \!=  \!\varnothing,  k,k'\! \in\! [t_0,t_0+\min(m,\lceil \alpha^{-1}\rceil)], k \neq k'.\!\!\!
\end{align}
Using this result, we will prove that the cardinality of the set $\bar{S}_k$ cannot be higher than $\mathcal{C}$ for more than $\lceil \alpha^{-1} \rceil$ timesteps. Let us assume by contradiction that $m > \lceil \alpha^{-1} \rceil$. Using \eqref{eq:SgammaRD} and \eqref{eq:Sgammanointersect} for $T \in [0,\lceil \alpha^{-1} \rceil]$, we get:
\begin{align*}
    \Bigg\lvert \bigcup_{k=t_0}^{t_0 + T} S^{\gamma}_k \Bigg\rvert = \sum_{k=t_0}^{t_0 + T} \lvert S^{\gamma}_k \rvert = (T+1) \mathcal{C}, 
\end{align*}
Fixing $T=\lceil \alpha^{-1} \rceil$, we obtain
\begin{align*}
    \Bigg\lvert \bigcup_{k=t_0}^{t_0 + \lceil k \rceil} S^{\gamma}(t) \Bigg\rvert = (\lceil \alpha^{-1} \rceil+1) \mathcal{C} \geq N + \mathcal{C} > N.
\end{align*}
Now, since the union takes care of duplications and the total number of agents in the game is $N$ for $T' \geq 0$,
\begin{align*}
    \Bigg\lvert \bigcup_{k=t_0}^{t_0 + T'} S^{\gamma}_k \Bigg\rvert \leq N,
\end{align*}
which leads to a contradiction. Thus, our assumption of $m > \lceil \alpha^{-1} \rceil$ was incorrect to begin with, and we have finally proved that $m \leq \lceil \alpha^{-1} \rceil$.

Next, we prove that $\tau_{\max,k} = \Os(\lceil \alpha^{-1} \rceil)$. For $k \notin [t_0,t_0+m]$, it is easy to see that $\tau_{\max,k+1} \leq \bar{\kappa}$. Moreover, for $k \in [t_0,t_0+m]$, $\tau_{\max,k+1} \leq \tau_{\max,k} + 1$. Now, since $m \leq \lceil \alpha^{-1} \rceil \leq \bar{\kappa}$, $\tau_{\max,k} \leq 2 \bar{\kappa}$ for any $k \geq 0$. Hence, the statement of the theorem follows, which completes the proof.

\section{Proof of Theorem \ref{Th:Asymptotic_optimality_deterministic}}\label{Ap:asymp_opt}

% \subsection{Proof of Theorem \ref{Th:Asymptotic_optimality_deterministic}}\label{appen_thm2}
%[Proof of Theorem \ref{Th:Asymptotic_optimality_deterministic}]
    Consider the following.
\begin{align}\label{derivation}
    & J^{BS}_{\hat{\gamma}} - J^{BS}_{\gamma_R}
    % = \limsup_{T \rightarrow \infty} \frac{1}{NT}\mathbb{E}\bigg[\sum_{k=0}^{T-1}\sum_{i=1}^N g(\Delta^i(t),A_i,K_{W^i})  + \omega(\tilde{\Delta}^i(t),A_i,K_{W^i})\bigg]\Bigg|_{\hat{\gamma}^d}  \nonumber \\
    % & \qquad \qquad \!\!-\limsup_{T \rightarrow \infty} \frac{1}{NT}\mathbb{E}\left[\sum_{k=0}^{T-1}\sum_{i=1}^N g(\Delta^i(t),A_i,K_{W^i})\right]\Bigg|_{\gamma^{d,*}_R} \nonumber \\
    \!=\! \overline{\lim\limits_{T \rightarrow \infty}} \frac{1}{NT}\mathbb{E}\!\!\left[\sum_{k=0}^{T-1}\sum_{i=1}^N
    c(\bar{\Delta},A_i,C_{W^i}) \right. \nonumber\\ & \left. \times \mathbf{1}_{\left\lbrace 1 \geq {\mathcal{C}}/{n^\lambda_k}\right\rbrace}\mathbf{1}_{\left\lbrace\tilde{\tau}^i_k \geq \kappa^i\right\rbrace}\right]  \nonumber \\
    & \leq \overline{\lim\limits_{T \rightarrow \infty}} \frac{\operatorname{U}}{NT}\mathbb{E}\left[\sum_{k=0}^{T-1}\sum_{i=1}^N
   \mathbf{1}_{\left\lbrace{n^\lambda_k}>{\mathcal{C}}\right\rbrace}\mathbf{1}_{\left\lbrace\tilde{\tau}^i_k \geq \kappa^i\right\rbrace} \right]  \nonumber \\
   & \leq \overline{\lim\limits_{T \rightarrow \infty}}\frac{\operatorname{U}}{T} \sum_{k=0}^{T-1}\mathbb{E}\left[
  \mathbf{1}_{\left\lbrace{n^\lambda_k}>{\mathcal{C}}\right\rbrace} \right]  \nonumber \\
   & = \overline{\lim\limits_{T \rightarrow \infty}}\frac{\operatorname{U}}{T} \sum_{k=0}^{T-1}\mathbb{P}(n^\lambda_k>\mathcal{C})  \nonumber \\
   & \leq \overline{\lim\limits_{T \rightarrow \infty}}\frac{\operatorname{U}}{T} \sum_{k=0}^{T-1}\mathbb{P}(e^{\theta n^\lambda_k}\geq e^{\theta\mathcal{C}}), ~\theta > 0 \nonumber \\
   & \leq \overline{\lim\limits_{T \rightarrow \infty}}\frac{\operatorname{U}}{T} \sum_{k=0}^{T-1}\inf_{\theta > 0}\frac{\mathbb{E}[e^{\theta n^\lambda_k}]}{e^{\theta\mathcal{C}}} \nonumber \\
    & = \overline{\lim\limits_{T \rightarrow \infty}}\frac{\operatorname{U}}{T} \sum_{k=0}^{T-1}\inf_{\theta \geq 0}\frac{(\mathbb{E}[e^{\theta a^i_k}])^N}{e^{\theta\mathcal{C}}} \leq \operatorname{U} \sum_{k=0}^{T-1} e^{-\operatorname{D}(\alpha||q)N},
    % & \leq \overline{\lim\limits_{T \rightarrow \infty}} \frac{\operatorname{U}}{NT}\mathbb{E}\left[\sum_{k=0}^{T-1}\left\lbrace {\Lambda^{\max}_k \!-\!\mathcal{C}} \right\rbrace^+ \right] \nonumber \\
    % & \leq \limsup_{T \rightarrow \infty} \frac{1}{NT}\mathbb{E}\left[\sum_{k=0}^{T-1}\sum_{i=1}^N \left\lbrace\frac{\Omega_k-R_d}{R_d} \right\rbrace^+\!\!\!\omega(\Delta^i(t),A_i,K_{W^i}) \right] \\
    % & \leq \limsup_{T \rightarrow \infty} \frac{1}{N^2T\alpha}\mathbb{E}\left[\sum_{k=0}^{T-1} |\Omega_k\!-\!R_d| \sum_{i=1}^N \omega(\Delta^i(t),A_i,K_{W^i})\! \right], \\
    % & \leq \frac{\operatorname{U}}{N\alpha}\overline{\lim\limits_{T \rightarrow \infty}} \frac{1}{T}\sum_{k=0}^{T-1}\mathbb{E}\left[ |\Omega_k-R_d|  \right] \nonumber \\
    % &  \leq \frac{\operatorname{U}}{N\alpha} \overline{\lim\limits_{T \rightarrow \infty}} \frac{1}{T}\sum_{k=0}^{T-1} \mathbb{E}\left[|\Lambda^{\max}_k - \mathbb{E}\left[ \Lambda \right]| \right], ~\text{w.p. $\geq 1-\delta$},
\end{align}
where $J^{BS}_{\hat{\gamma}}$ and $J^{BS}_{\gamma_R}$ are the costs under policies $\hat{\gamma}$ and $\gamma_R$, respectively, and $\operatorname{U}:= \max_{i \in [N]}c(\bar{\Delta},A_i,C_{W^i})$. The first equality follows since the sample paths of the AoI under $\hat{\gamma}$ coincide with those under the policy $\gamma_R$ by definition. The first inequality follows as a result of Proposition \ref{Prop:Delta_boundedness}. The third inequality follows by the monotonic nature of the exponential function and second-to-last inequality follows using Markov's inequality. Finally, the last equality follows because $a^i_k$'s are i.i.d. random variables (independent since they were computed using a decoupling procedure, and identically distributed since the probability of randomization $q$ is common for all agents for any given $k$). Finally, in the last inequality $\operatorname{D}(x||y):= x\ln{x/y} + (1-x)\ln{(1-x)/(1-y)}$ denotes the Kullback-Liebler divergence between independent Bernoulli distributed random variables distributed with parameters $x$ and $y$. Next, we observe that $\operatorname{D}(x||y) = 0$ if and only if $x=y$. For our case, this would then imply that $\alpha = q$, which can happen if and only if $\bar{\mathcal{C}} = 0$, which is not possible as a result of the constraint on $R(\lambda)$. Hence, $\operatorname{D}(x||y) >0$. Finally, by \eqref{Cost_comparison} and \eqref{derivation}, it follows that $J^{BS}_{\gamma} - J^{BS}_{\gamma_R} \leq J^{BS}_{\hat{\gamma}} - J^{BS}_{\gamma_R} \rightarrow 0$, exponentially fast, which completes the proof.

\section{Proof of Theorem \ref{Prop:Delta_High_deviation_bounded}}\label{appen_prop3}
%[Proof of Proposition \ref{Prop:Delta_High_deviation_bounded}]
    Let us start by defining two events $\tilde{\operatorname{E}}(\cdot)$ and $\hat{\operatorname{E}}(\cdot)$. The event $\tilde{\operatorname{E}}(x)$ is when any agent takes longer than $x$ time steps to re-enter the set $S^{\gamma}$ (which is the set of agents that need to be transmitted). The event $\hat{\operatorname{E}}(x)$ is when any agent takes longer than $x$ timesteps to be transmitted while in the set $S^{\gamma}$. Now, let us define the event $\bar{\operatorname{E}}(x)$ when any agent's AoI is larger than $x$ at any time instant. Then, we can deduce that
% \begin{align*}
%     \bar{\operatorname{E}}(x) \supseteq \big(\tilde{\operatorname{E}}(x) \cup \hat{\operatorname{E}}(x) \big),
% \end{align*}
% which upon using the union bound yields
% \begin{align} \label{eq:P_E_k}
%     \PP(\bar{\operatorname{E}}(x)) \geq \PP \big(\tilde{\operatorname{E}}(x) \cup \hat{\operatorname{E}}(k) \big) \geq  \min(\PP (\tilde{\operatorname{E}}(x)), \PP (\hat{\operatorname{E}}(x))).
% \end{align}
\begin{align*}
    \bar{\operatorname{E}}^c(x) & \supseteq \big(\tilde{\operatorname{E}}^c(x/2) \cap\hat{\operatorname{E}}^c(x/2) \big), \Leftrightarrow \\
    \bar{\operatorname{E}}(x) &\subset \big(\tilde{\operatorname{E}}(x/2) \cup\hat{\operatorname{E}}(x/2) \big),
\end{align*}
using which we deduce
\begin{align} \label{eq:P_E_k}
    \PP(\bar{\operatorname{E}}(x)) & \leq \PP \big(\tilde{\operatorname{E}}(x/2) \cup \hat{\operatorname{E}}(x/2) \big) \nonumber \\
    & \leq  \PP (\tilde{\operatorname{E}}(x/2)) + \PP (\hat{\operatorname{E}}(x/2)).
\end{align}
First analyze the event $\tilde{\operatorname{E}}(x)$. 
% \noindent \textbf{\color{blue}Melih you can input your proof here as is. I will connect the two proofs after you're done.}
The probability $\PP (\tilde{\operatorname{E}}(x))$ can be upper bounded as 
\begin{align} \label{eq:P_t_E}
   \PP (\tilde{\operatorname{E}}(x)) \leq \PP \bigg(\sum_{k = 0}^{x-1} | S_{suc,k}^{\gamma}|\leq N-\mathcal{C}-1 \bigg),
\end{align}
where $S_{suc,k}^{\gamma}$ denotes the set of agents that are in $S^{\gamma}_k$, and their updates are successfully transmitted at time $k$. In the worst case scenario, when $k\in [x]$, the number of agents whose updates are successfully transmitted is less than $N-\mathcal{C}-1$ so that in the $x^{th}$ time, the $i^{th}$ agent is still not present in the set $S^{\gamma}_x$. We note that $| S_{suc,k}^{\gamma}|$ is a Binomially distributed random variable with number of trials $\mathcal{C}$, and success probability $1-p$, i.e., $|S_{suc,k}^{\gamma}|\sim Bin(\mathcal{C}, 1-p )$, and $\PP(|S_{suc,k}^{\gamma}| = \ell) = {\mathcal{C} \choose \ell} (1-p)^{\ell} p^{\mathcal{C}-\ell}$, for $\ell\in[\mathcal{C}]$. Next, by using the fact that the sum of independent binomial distributions with the same success probabilities is also a Binomial distribution, we have that $\sum_{k = 0}^{x-1} | S_{suc,k}^{\gamma}| \sim Bin(x\mathcal{C}, 1-p)$. The probability of $\PP (\sum_{k = 0}^{x-1} | S_{suc,k}^{\gamma}|\leq N-\mathcal{C}-1 )$ can then be rewritten as
\begin{align}\label{probability_1}
    \!\!\PP \left( \frac{\sum_{x = 0}^{x-1} | S_{suc,k}^{\gamma}| - x \mathcal{C} (1-p)}{\sqrt{x\mathcal{C} p(1-p)}} \leq c_{x,N} \right), \!\!  
\end{align} 
where  
\begin{align*}
%c_{x,N}  =\frac{N-\mathcal{C}-1 - x \mathcal{C} (1-p) }{\sqrt{x\mathcal{C} p (1-p)}} \\
c_{x,N} = \frac{\sqrt{N}(1 - (1+x(1-p))\alpha) }{\sqrt{x \alpha p (1-p)}} - \frac{1}{\sqrt{x \alpha N p (1-p)}}. 
\end{align*}
Next, given $x > \big(\frac{2}{\alpha (1-p)}\big)^2$, we can obtain an upper bound on $c_{x,N}$ as:
\begin{align}
    c_{x,N} & \leq \frac{\sqrt{N}((1 - \alpha) - \alpha (1-p) x) }{\sqrt{x \alpha p (1-p)}} \nonumber \\
    & = \sqrt{\frac{N}{\alpha p (1-p)}} \bigg(\frac{1 - \alpha}{\sqrt{x }} - \alpha (1-p) \sqrt{x } \bigg) \nonumber \\
    & \leq \sqrt{\frac{N}{\alpha p (1-p)}} \bigg(\frac{(1- \alpha)\alpha (1-p)}{2} - 2  \bigg) \nonumber \\
    & \leq - \sqrt{\frac{N}{\alpha p (1-p)}} =: \bar{c}_N. \label{eq:c_N}
\end{align}
%\approx \sqrt{\frac{R_d}{kp(1-p)}}(\alpha-(k+1)p) $. 
Then, by using the central limit theorem, the distribution of $\frac{\sum_{k = 0}^{x-1} | S_{suc,k}^{\gamma}| - x \mathcal{C} (1-p)}{\sqrt{x\mathcal{C} p(1-p)}}$ converges to the standard Gaussian distribution, as the number $x\mathcal{C}$ gets large. Thus, by using the Berry–Esseen theorem \cite{berry_essen}, we have that 
\begin{align} \label{eq:Berry_Esseen}
   \left|\PP \left( \frac{\sum_{k = 0}^{x-1} | S_{suc,k}^{\gamma}| - x \mathcal{C} (1-p)}{\sqrt{x\mathcal{C} p(1-p)}} \leq c_{x,N} \right) - \Upphi(c_{k,N})\right|\leq \epsilon_1,
\end{align}
where $\Upphi(c_{x,N})=\frac{1}{\sqrt{2\pi}} \int_{-\infty}^{c_{x,N}} e^{-\frac{z^2}{2}} dz$ is the CDF of the standard Gaussian distribution and $\epsilon_1 \leq 0.33554 \frac{1-p +0.415}{\sqrt{x\mathcal{C}}}$ \cite[Theorem~2]{berry_essen}. Now, let us define two random variables:
\begin{align*}
    Z_N(x) & = \frac{\sum_{k = 0}^{x-1} | S_{suc,k}^{\gamma}| - x \mathcal{C} (1-p)}{\sqrt{x\mathcal{C} p(1-p)}}, \text{ and }\\
    Z_\infty(x) & = \lim_{N \rightarrow \infty} Z_N(x).
\end{align*}
Then, we know that the CDF of  $Z_\infty(x)$ is given by  $\Upphi(\cdot)$. Using \eqref{eq:P_t_E} and the definitions of $Z_N(x)$ and $Z_\infty(x)$, we get
\begin{align*}
    \PP (\tilde{\operatorname{E}}(x)) & \leq \PP(Z_N(x) \leq c_{x,N}) \\
    & \leq  \big\lvert \PP(Z_N(x) \leq c_{x,N}) - \PP(Z_\infty(x) \leq c_{x,N}) \big\rvert \\
    & \hspace{2cm} + \PP(Z_\infty(x) \leq c_{x,N}) \\
    & \leq 0.33554 \frac{1-p +0.415}{\sqrt{x\alpha N}} + \Upphi(c_{x,N}) \\
    & \leq 0.3354 \frac{1-p +0.415}{\sqrt{\alpha N}} + \Upphi(\bar{c}_{N}),
\end{align*}
where the third inequality follows from \eqref{eq:Berry_Esseen} and the fourth inequality follows from the fact that if $x > \big(\frac{2}{\alpha (1-p)}\big)^2$ then $c_{x,N} < \bar{c}_N$ using \eqref{eq:c_N}, which implies that $\Upphi(c_{x,N}) < \Upphi(\bar{c}_N)$ due to the monotonically increasing nature of $\Upphi(\cdot)$. Let us choose $N$ such that $0.3354 \frac{1-p +0.415}{\sqrt{\alpha N}} \leq \delta/4$ and $\Phi(\bar{c}_N) \leq \delta/4$. Then,
\begin{align} \label{eq:P_t_E_UB}
    \PP (\tilde{\operatorname{E}}(x)) \leq \delta/2.
\end{align}
Notice that the conditions on $N$ suggest a lower bound on $x$ which is independent of $N$ but dependent on $\delta$.

Next, we determine  $\PP (\hat{\operatorname{E}}^i(x))$, where $\hat{\operatorname{E}}^i(x)$ is the event that agent $i$ takes longer than $x$ timesteps to be transmitted while in the set $S^{\gamma}$. The former can be computed as:
\begin{align*}
    \PP (\hat{\operatorname{E}}^i(x)) = (1-p) \sum_{s=x}^\infty p^s = (1-p) \bigg( \frac{p^x}{1-p} \bigg) = p^x.
\end{align*}
Hence, if we choose $x \geq \log(2/\delta)/\log(1/p)$, then
\begin{align} \label{eq:P_h_E_UB}
    \PP (\hat{\operatorname{E}}^i(x)) \leq \delta/2,
\end{align}
Combining \eqref{eq:P_E_k}, \eqref{eq:P_t_E_UB} and \eqref{eq:P_h_E_UB} we get
\begin{align*}
     \PP (\bar{\operatorname{E}}^i(x)) \leq \delta
\end{align*}
for agent $i$, given that $N$ is chosen such that $0.3354 \frac{1-p +0.415}{\sqrt{\alpha N}} \leq \delta/4$ and $\Upphi(\bar{c}_N) \leq \delta/4$ and $x = \Os(\log(1/\delta))$. This then completes the proof of the theorem.

% Next, to determine  $\PP (\hat{\operatorname{E}}(x))$, we first define the event $\hat{\operatorname{E}}^i(x)$ which is when agent $i$ takes longer than $x$ timesteps to be transmitted while in the set $S^{\gamma}$. By definition $\hat{\operatorname{E}}(x) = \cap_{i=1}^N \hat{\operatorname{E}}^i(x)$ and so $\PP(\hat{\operatorname{E}}(x)) = \PP(\cap_{i=1}^N \hat{\operatorname{E}}^i(x))$. The probability $\PP (\hat{\operatorname{E}}^i(x))$ is quite easy to attain,
% \begin{align*}
%     \PP (\hat{\operatorname{E}}^i(x)) = (1-p) \sum_{s=x}^\infty p^s = (1-p) \bigg( \frac{p^x}{1-p} \bigg) = p^x.
% \end{align*}
% Hence, if we choose $x \geq \log(2/\delta)/\log(1/p)$ then $\PP (\hat{\operatorname{E}}^i(x)) \leq \delta/2$. Now using union bound we can determine that if $x \geq \log(2N/\delta)/\log(1/p)$ then 
% \begin{align} \label{eq:P_h_E_UB}
%     \PP (\hat{\operatorname{E}}(x)) \leq \delta/2
% \end{align}
% Combining \eqref{eq:P_E_k}, \eqref{eq:P_t_E_UB} and \eqref{eq:P_h_E_UB} we get
% \begin{align*}
%      \PP (\bar{\operatorname{E}}(x)) \leq \delta
% \end{align*}
% given $N$ is chosen such that $\frac{1-p +0.415}{\sqrt{\alpha N}} \leq \delta/4$ and $\Upphi(\bar{c}_N) \leq \delta/4$ and $x = \Os(\log(N/\delta))$.

\end{appendices}

\end{document}